\documentclass[preprint]{elsarticle}
\usepackage{hyperref}
\journal{Journal of TCS}

\mathcode`:="003A
\usepackage{color}
\usepackage{graphics}
\usepackage{amsmath}
\usepackage{amssymb}
\usepackage{latexsym}
\usepackage{diagrams}

\usepackage{subcaption}
\usepackage{pgfplots}
\usepackage{tikz}
\usetikzlibrary{arrows}
\usetikzlibrary{fadings}
\usepackage{xcolor}
\pgfplotsset{compat=1.12}

\newcommand{\defeq}{\stackrel{\vartriangle}{=}}

\newcommand{\rel}[3]{#2\,#1\,#3}
\newcommand{\Rf}[1]{\mathsf{Rf}_{#1}}
\newcommand{\Rs}[1]{\mathsf{Rs}_{#1}}
\newcommand{\Ef}[1]{\mathsf{Ef}_{#1}}
\newcommand{\Es}[1]{\mathsf{Es}_{#1}}
\newcommand{\Ts}{\mathsf{T}}
\newcommand{\Sp}{\mathsf{S}}

\newcommand{\defiff}{\stackrel{\vartriangle}{\iff}}

\newcommand{\hs}{\mathcal{H}}
\newcommand{\HS}{\mathbb{H}}

\newcommand{\Real}{\mathbb{R}}
\newcommand{\Nat}{\mathbb{N}}
\newcommand{\State}{\mathbb{S}}
\newcommand{\Time}{\mathbb{T}}

\newcommand{\Der}[1]{\dot{#1}}
\newcommand{\Next}[1]{{#1}^+}
\newcommand{\inv}[1]{{#1}^{-1}}
\newcommand{\bFa}[2]{\Box_{#1}(#2)}
\newcommand{\bca}[1]{#1^\Box}

\newcommand{\Den}[1]{[\![#1]\!]}
\newcommand{\X}{\times}

\newcommand{\PS}{\mathsf{P}}
\newcommand{\CS}{\mathsf{C}}
\newcommand{\OS}{\mathsf{O}}
\newcommand{\FS}[2]{|#1|(#2)}
\newcommand{\UpC}[1]{\uparrow\!{#1}}

\newcommand{\pL}{\mathbb{P}}
\newcommand{\cL}{\mathbb{C}}
\newcommand{\FL}[2]{#1(#2)}

\newcommand{\Set}{\mathbf{Set}}
\newcommand{\Po}{\mathbf{Po}}

\newcommand{\Top}{\mathbf{Top}}
\newcommand{\C}{\mathbb{A}}
\newcommand{\id}{\text{id}}

\newcommand{\cl}[1]{\overline{#1}}

\newdefinition{definition}{Definition}[section]
\newtheorem{proposition}[definition]{Proposition}
\newtheorem{lemma}[definition]{Lemma}

\newtheorem{theorem}[definition]{Theorem}
\newproof{proof}{Proof}
\newdefinition{notation}[definition]{Notation}
\newdefinition{remark}[definition]{Remark}
\newdefinition{example}[definition]{Example}

\begin{document}

\begin{frontmatter}
\title{Safe \& Robust Reachability Analysis of Hybrid Systems}
\author{Eugenio Moggi\fnref{genova}\footnote{Research partially supported by the Swedish Knowledge Foundation.}}
\author{Amin Farjudian\fnref{halmstad}}
\author{Adam Duracz\fnref{rice}\footnote{Work done while the author was a PhD student at Halmstad University.}}
\author{Walid Taha\fnref{halmstad}\footnote{Research partially supported by US NSF award \#1736754 ``A CPS Approach to Robot Design'' and the Swedish Knowledge Foundation project ``AstaMoCA: Model-based Communications Architecture for the AstaZero Automotive Safety Facility''.}}
\address[genova]{DIBRIS, Genova Univ., v. Dodecaneso 35, 16146 Genova, Italy, Email: moggi@unige.it}
\address[halmstad]{Halmstad University, Halmstad, Sweden, Email: name.surname@hh.se}
\address[rice]{Rice University, Houston, TX, USA, Email: adam.duracz@rice.edu}

\begin{abstract}
Hybrid systems---more precisely, their mathematical models---can exhibit behaviors, like \emph{Zeno behaviors}, that are absent in purely discrete or purely continuous systems.
First, we observe that, in this context, the usual definition of \emph{reachability}---namely, the reflexive and transitive closure of a transition relation---can be \emph{unsafe}, ie, it may compute a proper subset of the set of states \emph{reachable in finite time} from a set of initial states.
Therefore, we propose \emph{safe reachability}, which always computes a superset of the set of reachable states.

Second, in safety analysis of hybrid and continuous systems, it is important to ensure that a reachability analysis is also \emph{robust} wrt small perturbations to the set of initial states and to the system itself, since discrepancies between a system and its mathematical models are unavoidable.
We show that, under certain conditions, the \emph{best Scott continuous} approximation of an analysis $A$ is also its best robust approximation.
Finally, we exemplify the gap between the set of reachable states and the supersets computed by safe reachability and its best robust approximation.
\end{abstract}
\end{frontmatter}

\section*{Introduction}
In a transition system---ie, a binary relation $\to$ on a set of states---reachability is a clearly defined notion, namely, the reflexive and transitive closure $\to^*$ of $\to$.
Reachability analysis plays an important role in computer-assisted verification and analysis \cite{ACHH95},
since \textbf{safety} (a key system requirement) is usually formalized in terms of \textbf{reachability}, namely:
\begin{center}
state $s$ is safe $\iff$ it is not possible to reach a bad state from $s$.
\end{center}
For a hybrid system one can define a transition relation $\to$ on a \emph{continuous and uncountable} state space,
but $\to^*$ captures only the states reachable in finitely many transitions, and they can be a proper subset of those reachable in finite time!
Hybrid systems with \emph{Zeno behaviors}---where infinitely many events occur in finite time---are among the systems in which
the two notions of reachability differ.
Zeno behaviors arise naturally when modeling rigid body dynamics with impacts, as illustrated by the system consisting of a bouncing ball
(Example~\ref{ex-BB}), whose Zeno behavior is due to the modeling of impacts as discrete events.

\subsection*{Contributions}
The first contribution of this paper is the notion of \textbf{safe reachability} (Def~\ref{def-SR}), which
gives an over-approximation---ie, a superset---of the states reachable in finite time, including the case where the hybrid system has Zeno behaviors.
Mathematical models are always \emph{simplifications}, through abstractions and approximations, of \emph{real systems}.
Simplifications are essential to making analyses manageable.
In safety analysis, over-approximations are acceptable, since they can only lead to false negatives, ie,
the analysis may wrongly conclude that (a state $s$ of) the system in unsafe, because the over-approximation includes some unreachable bad states.

The second contribution is to show, under certain assumptions, that the \emph{best Scott continuous approximation} of safe reachability coincides with its \emph{best robust approximation}.
In safety analysis robust over-approximations are important, because inaccuracies in the modeling of a cyber-physical system (as well as in its building and testing) are unavoidable, as convincingly argued in \cite{franzle1999analysis}.

\subsection*{Related Work}

Reachability maps are arrows in the category of complete lattices and monotonic maps,
which is the standard setting for defining and comparing abstract interpretations \cite{cousot1992abstract}.
We build directly on the following papers.
\begin{itemize}
\item
\cite{GST2009} is an excellent tutorial on hybrid systems, from which we borrow the definition of hybrid system (Def~\ref{def-HS}),
but avoid the use of hybrid arcs, since they cannot reach nor go beyond \emph{Zeno points}.
\item
\cite{cuijpers2004,cuijpers2007} introduce topological transition
systems (TTS), which we use for defining safe reachability (Def~\ref{def-SR}).
In TTSs on discrete spaces, standard reachability (Def~\ref{def-R}) and safe
reachability (Def~\ref{def-SR}) coincide.
\item
\cite{Edalat95} is one among several papers, where Edalat recasts mainstream mathematics in Domain Theory,
and shows what is gained by doing so.
In this context Domain Theory becomes particularly relevant when the Scott and Upper Vietoris topologies on certain \emph{hyperspaces} coincide.

\item \cite{franzle1999analysis,gao2013delta} show that proving \emph{$\delta$-safety}, ie, safety of a hybrid system subject to some noise bounded by $\delta$, can make the verification task easier, besides excluding systems that are safe only under unrealistic assumptions.

\end{itemize}

\subsection*{Summary}

The rest of the paper is organized as follows:
\begin{itemize}
\item
Sec~\ref{sec-HS} recalls the definition of a hybrid system from \cite{GST2009}, defines the corresponding transition relation (Def~\ref{def-TR}), and gives some examples.
\item
Sec~\ref{sec-reach} introduces two reachability maps $\Rf{}$ and $\Rs{}$ (Def~\ref{def-R}~and~\ref{def-SR}, respectively),
establishes their properties and how they relate to each other.
\item
Sec~\ref{sec-approx} uses the poset-enriched category of complete lattices and monotonic maps (see Def~\ref{def-po-cat}) as a framework
to discuss approximations and relate reachability maps defined on different complete lattices.
We also give a systematic way to turn a monotonic map $f$ between complete lattices into a Scott continuous map $\bca{f}$ (see Prop~\ref{thm-bca}).
\item Sec~\ref{sec-robust} introduces the notion of robustness (see Def~\ref{def-robust}) and says when robustness and Scott continuity coincide (see Thm~\ref{thm-robust-main}).
\item Sec~\ref{sec-figs} analyses (with the aid of pictures) the differences between the under-approximation $\Rf{}$ and several over-approximations (from $\Rs{}$ to $\bca{\Rs{}}$) of sets of reachable states, for the hybrid systems introduced in Sec~\ref{sec-HS}.
\end{itemize}
\ref{sec-proofs} contains some proofs that were too long to inline into the text.

\section{Mathematical preliminaries}\label{sec-prelim}

We assume familiarity with the notions of Banach, metric, and topological spaces, and the definitions of open, closed, and compact subsets of a topological space (see, eg, \cite{Conway:Fun_Analysis:2ed:1990,Kel1975}).
The three notions of space are related as follows:
\begin{itemize}
    \item Every Banach space is a \emph{Cauchy complete} metric space whose distance is $d(x,y)\defeq|y-x|$, where $|x|$ is the norm of $x$;
    \item Every metric space is a topological space whose open subsets are given by unions of open balls $B(x,\delta)\defeq\{y|d(x,y)<\delta\}$.
\end{itemize}
Throughout the paper, for the sake of simplicity, one may replace Banach spaces with Euclidean spaces $\Real^n$.
For membership we may write $x:X$ instead of $x\in X$, and we use the following notations:
\begin{itemize}
\item $\Real$ is the Euclidean space of the real numbers;
\item $\Nat$ and $\omega$ denote the natural numbers with the usual linear order;
\item $\PS(\State)$ is the set of subsets of a set $\State$ (we use the same notation also when $\State$ is a set with additional structure, eg, a Banach or topological space);
\item $\OS(\State)$ is the set of open subsets of a topological space $\State$, and $\CS(\State)$ is the set of closed subsets (we use the same notation also when $\State$ is a set with additional structure that induces a topology, eg, a Banach or metric space);
\item $\Set$ is the category of sets and (total) maps;
\item $\Set_p$ is the category of sets and partial maps;
\item $\Top$ is the category of topological spaces and continuous maps.
\end{itemize}
Finally, we recall some definitions and their basic properties:
\begin{itemize}
\item $x$ is a \textbf{limit} of a sequence $(x_n|n:\omega)$ in the topological space $\State\defiff$\\
$\forall O:\OS(\State).x:O\implies\exists m.\forall n>m.x_n:O$.

The limits of a sequence form a closed subset of $\State$.
In a metric space a sequence has at most one limit.

\item $x$ is an \textbf{accumulation point} of a sequence $(x_n|n:\omega)$ in the topological space $\State\defiff$
$\forall O:\OS(\State).x:O\implies\forall m.\exists n>m.x_n:O$.

The accumulation points of a sequence form a closed subset of $\State$, every limit is also an accumulation point, and every accumulation point is a limit of a sub-sequence $(x_{f(n)}|n:\omega)$ for some strictly increasing $f:\Set(\omega,\omega)$, ie, $\forall n.f(n)<f(n+1)$.
In a metric space, if a sequence has a limit, then the limit is the only accumulation point.

\item The \textbf{derivative} $\Der{f}:\Set_p(\Real,\State)$ of a partial map $f:\Set_p(\Real,\State)$ from $\Real$ to a Banach space $\State$ is given by $\Der{f}(x)=v \defiff \exists \delta>0$ st $B(x,\delta)$ is included in the domain of $f$, and
for any sequence $(x_n|n:\omega)$ in $B(x,\delta)-\{x\}$, if $x$ is the limit of $(x_n|n:\omega)$ in $\Real$,
then $v$ is the limit of $(\frac{f(x_n)-f(x)}{x_n-x}|n:\omega)$ in $\State$.

If $\Der{f}(x)$ is defined, then $f$ must be continuous at $x$.  A stronger requirement is that $\Der{f}$ is defined and continuous in $B(x,\delta)$, in this case $f$ is called \emph{continuously differentiable} in $B(x,\delta)$.
\end{itemize}

\section{Hybrid Systems and Topological Transition Systems}\label{sec-HS}

In this section, we define what is a hybrid system (cf.~\cite{GST2009}), namely, a mathematical model suitable for describing cyber-physical systems;
what is a topological timed transition system (cf.~\cite{cuijpers2004}), namely, an abstraction of hybrid systems useful for defining
various reachability maps; and, finally, we introduce some example hybrid systems that will be used throughout the paper.

\begin{definition}\label{def-HS}
A \textbf{Hybrid System} (HS for short) $\hs$ on a Banach space $\State$ is a pair $(F,G)$ of binary relations on $\State$, ie, $F,G:\PS(\State\X\State)$, respectively called \textbf{flow} and \textbf{jump} relation.
We say that $\hs$ is \textbf{open/closed/compact}, when the relations $F$ and $G$---as subsets of the topological space $\State\X\State$---are open/closed/compact.
\end{definition}
\begin{remark}\label{rmk-HS}
In~\cite{GST2009}, the authors restrict $\State$ to a Euclidean space $\Real^n$, and
show that HS subsume Hybrid Automata \cite{ACHH95} and Switching Systems.
\end{remark}
The flow and jump relations are constraints for the \emph{trajectories} describing how the HS evolves over time
(see the definition of hybrid arc in \cite[page 39]{GST2009}).
Trajectories are not needed to define reachability (Sec~\ref{sec-reach}), since
a simpler notion of \emph{timed transition} suffices (see also \cite[Sec 2]{ACHH95} and \cite[Def 5]{P08}).
\begin{definition}\label{def-TR}
A \textbf{Topological Timed Transition System} (TTTS) is a pair $(\State,\rTo)$ consisting of
a topological space $\State$ and a timed transition relation $\rTo:\PS(\State\X\Time\X\State)$, where
$\Time\defeq\{d:\Real|d\geq 0\}$ is the continuous \textbf{time line}.
Its corresponding transition relation on $\State$ is given by $s\rTo s'\defiff\exists d.s\rTo^d s'$.

A HS $\hs=(F,G)$ on $\State$ induces a TTTS $(\State,\rTo_\hs)$ such that $s\rTo_\hs^d s'\defiff$
\begin{enumerate}
\item $d=0$ and $\rel{G}{s}{s'}$, or,
\item $d>0$ and there exists a continuous map $f:\Top([0,d],\State)$ such that:
\begin{itemize}
\item the derivative $\Der{f}$ of $f$ is defined and continuous in $(0,d)$;
\item $s=f(0)$, $s'=f(d)$ and $\forall t:(0,d).\rel{F}{f(t)}{\Der{f}(t)}$.
\end{itemize}
In this case we say that $f$ \emph{realizes} the transition.
\end{enumerate}
\end{definition}
\begin{remark}
The Banach space structure is just what is needed to define derivatives.
Hybrid arcs (cf. \cite{GST2009}) could be defined in term of a transition relation where the labels $d>0$ are replaced by their realizer maps $f$.

In \cite{GST2009} the requirements on $f$ are more relaxed than ours, namely: $f$ must be \emph{absolutely continuous} (which, in our case, is implied by the continuity of $\Der{f}$), and the flow relation must hold \emph{almost everywhere} in $(0,d)$.
However, the safe evolution and safe reachability maps (see Def~\ref{def-SR}) are insensitive to these changes. Thus, we have adopted the requirements on $f$ that are mathematically simpler to express.

In \cite{AdamPhD} the requirements on $f$ are stricter than ours, namely: $\Der{f}$ must extend continuously to $[0,d]$, and the flow relation must hold also at the end-points. For instance, the map $f(t)=\sqrt{t}$ is continuous in $[0,d]$, its derivative $\Der{f}(t)=\frac{1}{2*\sqrt{t}}$ is continuous in $(0,d)$, but it cannot be extended continuously to $0$.
The main rationale for the stricter requirements is that a transition $s\rTo^d s'$ with $d>0$ can only start from a state in the domain of the flow relation $F$.
\end{remark}
\begin{notation}
Given a first-order language with an interpretation in a Banach space $\State$, a HS on $\State$ can be described
by two formulas, a flow formula $F(x,\Der{x})$ and a jump formula $G(x,\Next{x})$, with two free variables each:
$x$ denotes the current state, $\Der{x}$ denotes the derivative of a trajectory flowing through $x$, and
$x^+$ denotes a state reachable from $x$ with one jump.

Similarly, given a two sorted language, with one sort interpreted in $\Real$ and the other in a topological space $\State$,
a timed transition relation can be described by a formula $T(x,d,x')$ with three free variables: $x$ denotes the starting state,
$d:\Real$ the duration of the transition, and $x'$ the final state.
\end{notation}
We introduce some hybrid systems, and give 
explicit descriptions of their timed transition relations (see also the Figures in Sec~\ref{sec-figs}).
\begin{example}[Expand]\label{ex-expand}
$\hs_E$ is a HS on $\Real$ describing the expansion of a quantity $m$ until it reaches a threshold $M>0$.
The flow and jump relations are:
\[F=\{(m,\Der{m})|0\leq m=\Der{m}\leq M\}, \qquad G=\emptyset.\]
It has two kinds of trajectories depending on the start state $m_0$ (see Fig~\ref{fig:expand-trajectories}).
\begin{enumerate}
\item When $m_0=0$, the quantity remains $0$ forever, ie, $f(t)=0$ when $t\geq 0$.
\item When $m_0:(0,M)$, there is an exponential growth $f(t)=m_0*e^{t}$ until $f(t)$ becomes $M$, then the trajectory cannot progress further.
\end{enumerate}
The timed transition relation $\rTo_{\hs_E}$ consists of the transitions
\begin{itemize}
\item $m\rTo^d m'$ with $0<d$ and $0\leq m\leq m'=m*e^d\leq M$.
\end{itemize}
Removing $(M,M)$ from $F$ does not change the timed transitions, while
adding $(M,0)$ to $F$ entails the addition of the transitions $M\rTo^d M$ with $d>0$.
\end{example}
\begin{example}[Decay]\label{ex-decay}
The hybrid system $\hs_D$ on $\Real$ describes the decay of a quantity when $m>0$, and a `refill' to $M>0$ when $m=0$. The flow and jump relations are:
\[F=\{(m,\Der{m})|m>0\land\Der{m}=-m\},\qquad G=\{(0,M)\}.\]
It has two kinds of trajectories, depending on the start state $m_0$ (see Fig~\ref{fig:decay-trajectories}):
\begin{enumerate}
\item When $m_0=0$, there is a refill followed by a decay $f(t)=M*e^{-t}$.
\item When $m_0>0$, there is a decay $f(t)=m_0*e^{-t}$. Thus, $f(t)>0$ when $t\geq 0$, and $f(t)\to 0$ as $t\to+\infty$.
\end{enumerate}
The timed transition relation $\rTo_{\hs_D}$ consists of the transitions
\begin{itemize}
\item $0\rTo^0 M$, and
\item $m\rTo^d m'$ with $0<d$ and $m>m'=m*e^{-d}>0$.
\end{itemize}
Adding $(0,0)$ to $F$ entails the addition of the transitions $0\rTo^d 0$ with $d>0$.
\end{example}
\begin{example}[Bouncing ball]\label{ex-BB}
The hybrid system $\hs_B$ on $\Real^2$ describes a bouncing ball with height $h\geq 0$ and
velocity $v$, which is kicked when it stops, ie, when $h=v=0$.
We assume the force of gravity to be $-1$ (for the sake of simplicity), a coefficient of restitution $b$ (we do not restrict its value, but $b:[-1,0]$ would be the obvious restriction), and a velocity $V>0$ given to the ball when it is kicked.
Formally:
\begin{itemize}
    \item $F=\{((h,v),(\Der{h},\Der{v}))|h>0\land\Der{h}=v\land\Der{v}=-1\}$.
    \item $G=\{((0,v),(0,v^+))|v<0\land v^+=b*v\}\uplus\{((0,0),(0,V))\}$.
\end{itemize}
It has seven kinds of trajectories starting from $(h=0,v>0)$,
depending on the value of $b$ (see Fig~\ref{fig:bb-trajectories}).
\begin{enumerate}
\item When $b<-1$, the ball never stops (its energy increases at each bounce).
\item When $b=-1$, the ball never stops (its energy remains constant).
\item When $b:(-1,0)$,
the ball stops in finite time, but after infinitely many bounces (this is a \textbf{Zeno behavior}), then it is kicked, ie, $(h=0,v=V)$.
\item When $b=0$, the ball stops as it hits the ground, then it is kicked.
\item When $b:(0,1)$, as the ball hits the ground, it stops after infinitely many instantaneous slowdowns $0>b^n*v\to 0$ (this is a \textbf{chattering Zeno behavior}), then it is kicked.
\item When $b=1$, as the ball hits the ground, the trajectory cannot progress further in time.
\item When $b>1$, as the ball hits the ground, its velocity drifts to $-\infty$ after an
  infinite sequence of instantaneous accelerations $0>b^n*v\to -\infty$, and the trajectory cannot progress further in time.
\end{enumerate}
The timed transition relation $\rTo_{\hs_B}$ consists of the following transitions:
\begin{itemize}
\item $(0,v)\rTo^0 (0,v')$ with $v<0$ and $v'=b*v$, this is a bounce;
\item $(0,0)\rTo^0 (0,V)$, this is when the ball is kicked;
\item $(h,v)\rTo^d (h',v')$, with $0<d$ and $0\leq h,h'=h+v*d-\frac{d^2}{2}\land v'=v-d$, this is
when the ball moves while the energy $E(h,v)=h+\frac{v^2}{2}$ stays constant.
\end{itemize}
In particular,  $(0,v)\rTo^{2*v} (0,-v)$ is the transition between two bounces.
Adding $\{((0,v),(\Der{h},\Der{v}))|\Der{h}=v\land\Der{v}=-1\}$ to $F$ does not change the timed transitions, while adding $((0,0),(0,0))$ to $G$ entails the addition of $(0,0)\rTo^0 (0,0)$.
\end{example}
The following construction adds a clock to a HS to record the passing of time.
\begin{definition}\label{def-HT}
Given a HS $\hs=(F,G)$ on $\State$,
the derived HS $t(\hs)=(F',G')$ on $\Real\X\State$
\emph{adds} a clock to $\hs$, namely:
\begin{itemize}
\item  $F'\defeq\{((t,s),(1,\Der{s}))|\rel{F}{s}{\Der{s}}\}$, because $dt/dt = 1$;
\item $G'\defeq\{((t,s),(t,s^+))|\rel{G}{s}{s^+}\}$, because jumps are instantaneous.
\end{itemize}
\end{definition}
\begin{proposition}\label{thm-HT}
$(t,s)\rTo_{t(\hs)}^d(t',s')\iff t'=t+d\land s\rTo_\hs^d s'$.
\end{proposition}
\begin{proof}
Only the case $d>0$ is non-trivial.
\begin{itemize}
\item If $f':\Top([0,d],\Real\X\State)$ realizes the transition $(t,s)\rTo_{t(\hs)}^d(t',s')$, then
$f'(x)=(t+x,f(x))$ for a unique $f:\Top([0,d],\State)$, and moreover this $f$ realizes the transition $s\rTo_\hs^d s'$.
\item Conversely, if $f:\Top([0,d],\State)$ realizes $s\rTo_\hs^d s'$, then
$f'(x)=(t+x,f(x))$ realizes $(t,s)\rTo_{t(\hs)}^d(t',s')$.\qed
\end{itemize}
\end{proof}

\section{Evolution and Reachability}\label{sec-reach}
Transition systems (TS for short) provide the main formalism for modeling discrete systems.
The formalism does not mention time explicitly, but it assumes that time is discrete, and each transition takes one time unit (or alternatively, it abstracts from time and describes only the order of discrete state changes).

Given a TS $(\State,\to)$, ie, a binary relation $\to$ on a set $\State$ (aka a directed graph),
we identify the \emph{discrete time line} with the set $\Nat$ of natural numbers and define the following notions related to the TS.
\begin{itemize}
\item A trajectory is a map $f:\Set([0,n],\State)$ such that $\forall i<n.f(i)\to f(i+1)$ for some $n:\Nat$,
or equivalently, a path $f:\State^+$ in the graph. The length of $f$ is $n$ and $f(0)$ is its starting state.
\item The evolution map $\Ef{}:\PS(\State)\to \PS(\Nat\X\State)$ is $\Ef{}(I) \defeq \{(n,s')|\exists s:I.s\to^n s'\}$, or equivalently the union of (the graphs of) all trajectories starting from the set $I$ of initial states.
Therefore, $\Ef{}(I)$ says at what time a state is reached, but forgets the trajectories used to reach it.
However, when $\to$ is deterministic, there is at most one trajectory of length $n$ from $s$ to $s'$,
which can be recovered from $\Ef{}(\{s\})$.
\item The reachability map $\Rf{}:\PS(\State)\to \PS(\State)$ is $\Rf{}(I) \defeq \{s'|\exists s:I.s\to^* s'\}$, or equivalently $\{s'|\exists n.(n,s'):\Ef{}(I)\}$.
Therefore, $\Rf{}(I)$ says whether a state is reachable from $I$, but forgets at which time instances it is reached.
\end{itemize}
For TTTS (and HS) one would like to reuse as much as possible the theory available for TS. The main point of this section is that naive reuse can result under-approximating what is reachable in finite time. To address this problem, we present a solution that computes an over-approximation (see Sec~\ref{sec-inclusions}).  This solution exploits
the topological structure of the state space $\State$ and the continuous time line $\Time$.

We choose to cast analyses (eg reachability) as monotonic maps (like $\Rf{}$)  rather than as relations (like $\to^*$). This becomes essential in Def~\ref{def-SR} and for defining approximability (Sec~\ref{sec-approx}) and robustness (Sec~\ref{sec-robust}) of an analysis.

\begin{definition}\label{def-R}
The \textbf{evolution} map $\Ef{}:\PS(\State)\to \PS(\Time\X\State)$ and
the \textbf{reachability} map $\Rf{}:\PS(\State)\to \PS(\State)$ for a TTTS $(\State,\rTo)$ are:
\begin{itemize}
\item $\Ef{}(I)\defeq$ the smallest $S:\PS(\Time\X\State)$ such that $\{0\}\X I\subseteq S$ and
$S$ is \emph{closed} wrt timed transitions, ie, $(t,s)\in S\land s\rTo^d s'\implies (t+d,s')\in S$.
\item $\Rf{}(I)\defeq$ the smallest $S:\PS(\State)$ such that $I\subseteq S$ and\\
$S$ is \emph{closed} wrt transitions, ie, $s\in S\land s\rTo s'\implies s'\in S$.
\end{itemize}
We denote with $\Ef{\hs}$ and $\Rf{\hs}$ the evolution and reachability maps for the TTTS induced by the HS $\hs$.
\end{definition}
\begin{remark}
The "f" in $\Ef{}$ and $\Rf{}$ stands for "finite", because these maps consider only states that are reachable in \textbf{finitely many} transitions.
There is an important difference between discrete systems and continuous/hybrid systems.
In a discrete (time) system the transition relation suffices to define trajectories, the evolution, and the reachability maps.
In a continuous (time) system: to define trajectories, the structure of a HS is needed; to define the evolution map, the timed transition relation suffices; and to define the reachability map, the transition relation suffices.
\end{remark}
\begin{theorem}\label{thm-fRE}
The following properties hold:
\begin{enumerate}
\item $\Ef{}$ is monotonic, ie, $I_0\subseteq I_1\implies\Ef{}(I_0)\subseteq\Ef{}(I_1)$, and preserves unions, ie,
$\forall K\subseteq\PS(\State).\Ef{}(\cup K)=\cup\{\Ef{}(I)|I:K\}$.
\item $\Rf{}$ is monotonic, preserves unions, is a closure, ie, $I\subseteq \Rf{}(I)=\Rf{}^2(I)$,
and $\pi(\Ef{}(I))=\Rf{}(I)$.
\item
If $\hs$ is a HS on $\State$, then
$\forall I:\PS(\State).\Ef{\hs}(I)=\Rf{t(\hs)}(\{0\}\X I)$ and\\
$\forall J:\PS(\Real\X\State).\pi(\Rf{t(\hs)}(J))=\Rf{\hs}(\pi(J))$.
\end{enumerate}
Here, $\pi:\Real\X\State\rTo\State$ is $\pi(t,s)\defeq s$,
and $\pi(J)$ is the image of $J\subseteq\Real\X\State$.
\end{theorem}
\begin{proof}See \ref{sec-proofs}.\end{proof}
A pair $(t,s')$ is in $\Ef{\hs}(I)$ exactly when $s'$ is reached at time $t$ in finitely many transitions starting from some $s:I$.
In \emph{Zeno systems} there are states that are reached at a finite time $t$, but not in a finite number of transitions.
Therefore, $\Ef{\hs}$ and $\Rf{\hs}$ may under-approximate what we would like to compute.

\begin{definition}\label{def-Zeno}
A \textbf{Zeno behavior} of $\hs$ is a sequence $((d_n,s_n)|n:\omega)$ in $\Time\X\State$ such that:
\begin{enumerate}
\item $\forall n.s_n\rTo_\hs^{d_n}s_{n+1}$,
\item $d\defeq\displaystyle\sum_{n:\omega} d_n$ is defined and finite, and,
\item the sequence has infinitely many jumps, ie, $\{n|d_n=0\}$ is infinite.
\end{enumerate}
The last requirement excludes \emph{fragmentations} of a flow transition $s\rTo{\hs}^d s'$,
ie, sequences $((f(t_n),t_{n+1}-t_n)|n:\omega)$, where
$f:\Top([0,d],\State)$ realizes $s\rTo{\hs}^d s'$, and
$(t_n|n:\omega)$ is a strictly increasing sequence with $t_0=0$ and $\sup_{n:\omega} t_n=d$.

The accumulation points of $(s_n|n:\omega)$ in the topological space $\State$ are called the \textbf{Zeno points},
and $d$ is called the \textbf{Zeno time}, since it is the time needed to reach a Zeno point from $s_0$.
\end{definition}
$\hs_B$ of Example~\ref{ex-BB} is the classical case of a HS with Zeno behavior.
When $b$ is in the interval $(-1,0)$, the stop state $s=(0,0)$ is reached 
in finite time from $s_0=(0,v)$ with $v<0$, but
after infinitely many bounces (see Fig~\ref{fig:bb-trajectories} in Sec~\ref{sec-figs}).
When $b$ is in the interval $(0,1)$, $\hs_B$ has a chattering Zeno behaviour, ie,
the stop state is reached after infinitely many instantaneous slowdowns.
On the other hand, Prop~\ref{thm-BB1} shows that the stop state is not in $\Rf{\hs_B}(\{s_0\})$ when $b\not=0$.
\begin{proposition}\label{thm-BB1}
Let $S\defeq\{(h,v)|h\geq 0\land E(h,v)>0\}$, where $E(h,v)=h+\frac{v^2}{2}$ is
the energy in state $(h,v)$. Then, $\Rf{\hs_B}(S)=S$, provided that $b\not=0$.
\end{proposition}
\begin{proof}
We prove that $S$ is closed wrt the transition relation $\rTo_{\hs_B}$ by case analysis. There are three cases:
\begin{itemize}
\item bounce, ie, $(0,v)\rTo (0,b*v)$ with $v<0$: $(0,v)\in S$, because $E(0,v)=\frac{v^2}{2}>0$, and $(0,b*v)\in S$, because $E(0,b*v)>0$ when $b\not=0$;
\item ball kicked, ie, $(0,0)\rTo (0,V)$: there is nothing to prove, as $(0,0)\not\in S$;
\item move, ie, $(h,v)\rTo (h',v')$ with $h,h'\geq 0\land v>v'\land E(h,v)=E(h',v')$: $(h,v)\in S\implies (h',v')\in S$ holds, because the energy stays constant. \qed
\end{itemize}
\end{proof}
We propose a key change to Def~\ref{def-R} that exploits the topology on $\State$ and $\Time$
by considering reachable also a state that is \emph{arbitrarily close} to reachable states.
\begin{definition}[Safe maps]\label{def-SR}
Let $\CS(\State)$ be the set of \textbf{closed subsets} of a topological space $\State$.
The \textbf{safe evolution} map $\Es{}:\CS(\State)\to \CS(\Time\X\State)$ and
the \textbf{safe reachability} map $\Rs{}:\CS(\State)\to \CS(\State)$ for a TTTS $(\State,\rTo)$ are:
\begin{itemize}
\item $\Es{}(I)\defeq$ the smallest $S:\CS(\Time\X\State)$ such that $\{0\}\X I\subseteq S$ and
\emph{closed} wrt timed transitions.
\item $\Rs{}(I)\defeq$ the smallest $S:\CS(\State)$ such that $I\subseteq S$, and
\emph{closed} wrt transitions.
\end{itemize}
We denote with $\Es{\hs}$ and $\Rs{\hs}$ the safe evolution and safe reachability maps for the TTTS induced by the HS $\hs$, respectively.
\end{definition}
\begin{remark}
If $s_0:I$ and $((d_n,s_n)|n:\omega)$ is a Zeno behavior of $\hs$, then the set $S=\Rs{\hs}(I)$
includes every Zeno point $s$ (similarly $\Es{\hs}(I)$ includes $(d,s)$, where $d$ is the Zeno time and $s$ is a Zeno point).
In fact, the sequence $(s_n|n:\omega)$ is included in $S$, and all its accumulation points must be in $S$, because $S$ is closed.
The set $S$ also includes asymptotically reachable points---ie, 
accumulation points of a sequence $(s_n|n:\omega)$ such that
$\forall n.s_n\rTo_\hs^{d_n}s_{n+1}$ and $\sum_{n:\omega} d_n=+\infty$---that may not be reachable in finite time.

The safe maps include other points that should be considered reachable in finite time,
but are not reachable in a finite number of transitions.  For instance, consider a
HS $\hs=(F,\emptyset)$ on $\Real^n$ that can only flow---thus, it cannot have Zeno behaviors---and a continuous map $f:\Top([0,d),\Real^n)$ such that:
\begin{itemize}
\item the derivative $\Der{f}$ of $f$ is defined and continuous in $(0,d)$,
\item $\forall t:(0,d).\rel{F}{f(t)}{\Der{f}(t)}$, but,
\item there is no way to extend $f$ to a continuous map on $[0,d]$.
\end{itemize}
If $f(0):I$, then $\Es{\hs}(I)$ includes $\{(t,f(t))|t<d\}$ and also the pairs $(d,s)$
with $s$ accumulation point of $f(t)$ as $t\to d$.
For instance, $f(x)=x*\sin(\frac{1}{x-1})$ satisfies the above properties for $d=1$.
\end{remark}
There is an analogue of Thm~\ref{thm-fRE} for the safe maps, but with weaker properties,
mainly because the set of closed subsets is closed only wrt finite unions.
\begin{theorem}\label{thm-sRE}
The following properties hold:
\begin{enumerate}
\item $\Es{}$ is monotonic and preserves finite unions.
\item $\Rs{}$ is monotonic, preserves finite unions, is a closure,
and\\$\pi(\Es{}(I))\subseteq\Rs{}(I)$.
\item
$\forall I:\PS(\State).\Ef{}(I)\subseteq\Es{}(\cl{I})$, and
$\forall I:\PS(\State).\Rf{}(I)\subseteq\Rs{}(\cl{I})$.
\item
If $\hs$ is a HS on $\State$, then
$\forall I:\CS(\State).\Es{\hs}(I)=\Rs{t(\hs)}(\{0\}\X I)$, and\\
$\forall J:\CS(\Real\X\State).\pi(\Rs{t(\hs)}(J))\subseteq\Rs{\hs}(\cl{\pi(J)})$.
\end{enumerate}
Here, $\cl{S}:\CS(\State)$ is the \textbf{closure} of $S:\PS(\State)$, ie,
the smallest $S':\CS(\State)$ such that $S\subseteq S'$.
\end{theorem}
\begin{proof} See \ref{sec-proofs}.\end{proof}

\subsection{Summary of inclusion relations}\label{sec-inclusions}

We provide a summary of the inclusion relations among the sets computed by the four maps defined in this section.
Given a hybrid system $\hs$ on $\State$ and a subset $I:\PS(\State)$ of initial states,
there are two subsets $E:\PS(\Time\X\State)$ and $R:\PS(\State)$ \emph{informally defined} as:
\begin{itemize}
\item $E$: the set of $(t,s)$ such that $s$ is reached at time $t$, ie, there is a \emph{trajectory} of $\hs$ starting from a state in $I$ and reaching $s$ at time $t$.
\item $R$: the set of states reachable (from $I$) in finite time, ie, $R=\pi(E)$.
\end{itemize}
The monotonic maps in Def~\ref{def-R}~and~\ref{def-SR} allow to define four subsets:
\begin{itemize}
\item $\Ef{\hs}(I)$: the set of $(t,s)$ such that $s$ is reached at time $t$ in finitely many transitions.
\item $\Rf{\hs}(I)$: the set of states reachable in finitely many transitions, ie, $\pi(\Ef{\hs}(I))$.
\item $\Es{\hs}(\cl{I}):\CS(\Time\X\State)$ a closed over-approximation of $E$.
\item $\Rs{\hs}(\cl{I}):\CS(\State)$ a closed over-approximation of $R$.
\end{itemize}
When $I:\CS(\State)$, the inclusion relations among these six subsets are:
\[\commdiag{
\Ef{\hs}(I)&\rIntoA&E\\
\dIntoA&&\dIntoA\\
\cl{\Ef{\hs}(I)}&\rIntoA&\Es{\hs}(I)\\
}\qquad\commdiag{
\Rf{\hs}(I)&\rIntoA&R&\rIntoA&\pi(\Es{\hs}(I))\\
\dIntoA&&&&\dIntoA\\
\cl{\Rf{\hs}(I)}&&\rIntoA&&\Rs{\hs}(I)\\
}\] 
By suitable choices of closed HS $\hs$ on $\Real$ and singletons $I$ we show that no other inclusion holds.
In particular, $\Ef{\hs}$ and $\Rf{\hs}$ may compute proper under-approximations,
while $\Es{\hs}$ and $\Rs{\hs}$ may compute proper over-approximations.
\begin{enumerate}
\item
$\Ef{\hs}(I)\subset\cl{\Ef{\hs}(I)}\subset E=\Es{\hs}(I)$ and
$\Rf{\hs}(I)\subset\cl{\Rf{\hs}(I)}\subset R=\Rs{\hs}(I)$.\\
Take $\hs=(\emptyset,G)$ with $G\defeq\{(x,x/2)|0\leq x\}\uplus\{(0,2)\}$ and $I=\{1\}$, then
\begin{itemize}
\item $\Rf{\hs}(I)=\{2^{-n}|n:\Nat\}$ and
$\Ef{\hs}(I)=\{0\}\X\Rf{\hs}(I)$
\item $\cl{\Rf{\hs}(I)}=\{2^{-n}|n:\Nat\}\uplus\{0\}$ and
$\cl{\Ef{\hs}(I)}=\{0\}\X\cl{\Rf{\hs}(I)}$
\item $\Rs{\hs}(I)=\{2^{-n}|n:\Nat\}\uplus\{0,2\}$ and
$\Es{\hs}(I)=\{0\}\X\Rs{\hs}(I)$
\end{itemize}

\item $\Ef{\hs}(I)=E\subset\cl{\Ef{\hs}(I)}\subset\Es{\hs}(I)$ and
$\Rf{\hs}(I)=R\subset\cl{\Rf{\hs}(I)}\subset\pi(\Es{\hs}(I))$.\\
Take $\hs=(\emptyset,G)$ with $G\defeq\{(2,x)|x\geq 2\}\uplus\{(x,1/x)|x\geq 2\}\uplus\{(0,1)\}$ and $I=\{2\}$, then
\begin{itemize}
\item $\Rf{\hs}(I)=[2,+\infty)\uplus (0,1/2]$ and
$\Ef{\hs}(I)=\{0\}\X\Rf{\hs}(I)$
\item $\cl{\Rf{\hs}(I)}=[2,+\infty)\uplus [0,1/2]$ and
$\cl{\Ef{\hs}(I)}=\{0\}\X\cl{\Rf{\hs}(I)}$
\item $\Rs{\hs}(I)=[2,+\infty)\uplus [0,1/2]\uplus\{1\}$ and
$\Es{\hs}(I)=\{0\}\X\Rs{\hs}(I)$
\end{itemize}

\item $\Rf{\hs}(I)=R=\pi(\Es{\hs}(I))\subset\cl{\Rf{\hs}(I)}\subset\Rs{\hs}(I)$.\\
Take $\hs=(F,\{(0,2)\})$ with $F\defeq\{(x,-x)|0\leq x\}$ and $I=\{1\}$, then
\begin{itemize}
\item $\Ef{\hs}(I)=E=\Es{\hs}(I)=\{(t,e^{-t})|0\leq t\}$
\item $\Rf{\hs}(I)=(0,1]$
\item $\cl{\Rf{\hs}(I)}=[0,1]$
\item $\Rs{\hs}(I)=[0,2]$
\end{itemize}
\end{enumerate}

\section{A Framework for Approximability}\label{sec-approx}

All maps introduced in Sec~\ref{sec-reach} (see Def~\ref{def-R}~and~\ref{def-SR}) are monotonic maps between complete lattices. Thus, they \emph{live} in the \emph{poset-enriched} category $\Po$ of complete lattices (more generally of posets) and monotonic maps.
In fact, one can stay within $\Po$, by defining a complete lattice of hybrid systems on $\State$ (this can be done also for timed transition relations), by exploiting its cartesian closed structure, and by using the order relation to express when something is an over- or under-approximation of something else.
$\Po$ is also the natural setting for defining and comparing \emph{abstract interpretations} \cite{cousot1992abstract}, with \emph{adjunctions} giving a systematic way to relate more concrete to more abstract interpretations.

Therefore, $\Po$ will be used as a framework in which to place and compare the reachability (and evolution) maps introduced so far, as well as their variants.
\begin{definition}\label{def-po-enriched}
A \textbf{poset-enriched} category $\C$ (see \cite{kelly1982basic}) consists of:
\begin{itemize}
\item a class of objects, notation $X:\C$;
\item a \textbf{poset} $\C(X,Y)$ of arrows from $X$ to $Y$, notation  $X\rTo^f Y$ and $f_1\leq f_2$;
\item identities $X\rTo^{\id_X} X$ and composition $X\rTo^{g\circ f}Z$ of composable arrows $X\rTo^f Y\rTo^g Z$
satisfying the usual equations and \textbf{monotonicity}, ie,\\
$f_1\leq f_2:X\to Y\land g_1\leq g_2:Y\to Z\implies (g_1\circ f_1)\leq(g_2\circ f_2):X\to Z$.
\end{itemize}
An \textbf{adjunction} $f\dashv g$ is a pair $X\pile{\rTo^f\\\lTo_g}Y$ of maps st
$f\circ g\leq\id_Y\land\id_X\leq g\circ f$.
If $\C$ has a terminal object $1$ and $h:\C(X,X)$, then:
\begin{itemize}
\item $\mu:\C(1,X)$ is an \textbf{initial algebra} (aka least prefix-point) for $h$ $\defiff$\\
$h\circ \mu\leq \mu$ and $\forall x:\C(1,X).h\circ x\leq x\implies \mu\leq x$.
\item $\nu:\C(1,X)$ is a \textbf{final co-algebra} (aka greatest postfix-point) for $h$ $\defiff$\\
$\nu\leq h\circ \nu$ and $\forall x:\C(1,X).x\leq h\circ x\implies x\leq\nu$.
\end{itemize}
\end{definition}
\begin{remark}
The terminology in Def~\ref{def-po-enriched} comes from Category Theory, in particular from \emph{2-categories},
since posets are a degenerate form of categories.  Ordinary categories amount to a poset-enriched categories where the order on $\C(X,Y)$ is equality.
If $\C$ is poset-enriched, then $\C^{op}$ denotes $\C$ with the direction of arrows reversed, ie, $\C^{op}(X,Y)=\C(Y,X)$, and $\C^{co}$ denotes $\C$ with the order on $\C(X,Y)$ reversed, ie, $\C^{co}(X,Y)=\C(X,Y)^{op}$.  Any notion in $\C$ has a dual notion in $\C^{op}$ and a co-notion in $\C^{co}$, in particular:
\begin{itemize}
\item $f\dashv g$ in $\C$ $\iff$ $g\dashv f$ in $\C^{co}$.
\item $x$ initial algebra for $h$ in $\C$ $\iff$ $x$ final co-algebra for $h$ in $\C^{co}$.
\end{itemize}
\end{remark}
The following facts are instances of more general results valid in 2-categories.
\begin{proposition}\label{thm-fR-unique}
In any poset-enriched category $\C$:
\begin{enumerate}
    \item If $f\dashv g$, then $g$ is uniquely determined by $f$, and when $g$ exists
    it is denoted by $f^R$, and called the \textbf{right adjoint} to $f$.
    \item If $f:X\to Y$ is an isomorphism, then $f\dashv\inv{f}$.
\end{enumerate}
\end{proposition}
\begin{definition}\label{def-po-cat}
The poset-enriched category $\Po$ is defined as follows:
\begin{itemize}
\item Objects $X:\Po$ are \textbf{complete lattices}, ie,
posets $(|X|,\leq_X)$ such that every subset $S$ of $|X|$ has a \textbf{sup}, denoted as $\bigsqcup S$. In paricular, $\bigsqcup \emptyset$ is the least element $\bot_X$ of $X$.
\item Arrows $f:\Po(X,Y)$ are \textbf{monotonic} maps $f:|X|\to|Y|$, ie,\\
$\forall x_0,x_1:|X|.x_0\leq_X x_1\implies f(x_0)\leq_Y f(x_1)$ (sups are irrelevant).
\item The poset-enrichment $f\leq g:\Po(X,Y)$ is given by the point-wise order.
\end{itemize}
\end{definition}
\begin{remark}
$\Po$ is \emph{cartesian closed} in the enriched sense, ie, it has a terminal object and the following poset isomorphisms \emph{natural} in $Z$:
\begin{itemize}
\item
$\Po(Z,X\X Y)\cong\Po(Z,X)\X\Po(Z,Y)$, where $X\X Y$ is the cartesian product of the complete lattices $X$ and $Y$.
\item
$\Po(Z,Y^X)\cong\Po(Z\X X,Y)$, where $Y^X$ is the complete lattice $\Po(X,Y)$.
\end{itemize}
Moreover, $f\dashv g$ in $\Po$ exactly when $f$ and $g$ form a \emph{Galois connection}, ie:
\[\forall x:|X|.\forall y:|Y|.x\leq_X g(y)\iff f(x)\leq_Y y.\]
Restricting the objects of $\Po$ to complete lattices allows to characterize the arrows $f:\Po(X,Y)$ that have right adjoints, and implies that every $h:\Po(X,X)$ has an initial algebra and a final co-algebra (the proofs of Thm~\ref{thm-fRE}~and~\ref{thm-sRE} make systematic use of the universal property of initial algebras).
\end{remark}
\begin{proposition}\label{thm-fR-galois}
Given a map $f:\Po(X,Y)$ the following are equivalent:
\begin{itemize}
\item $f$ preserves all sups, ie, $f(\bigsqcup S)=\bigsqcup(f\,S)$ when $S\subseteq|X|$.
\item $f$ has a right adjoint $f^R$, namely, $f^R(y)=\bigsqcup\{x|f(x)\leq_Y y\}$.
\end{itemize}
\end{proposition}
For a HS on $\State$, the complete lattices $X$ of interest should have as underlying set a subset of $\PS(\State)$, but there are two choices for $\leq_X$: inclusion $\subseteq$; and reverse inclusion $\supseteq$.
When $\leq_X$ is an \emph{information order}, the correct choice is reverse inclusion, since
a smaller \emph{over-approximation} is more informative.
\begin{definition}\label{def-lattices}
Given a topological space $\State$, define:
\begin{itemize}
\item
$\pL(\State)\defeq(\PS(\State),\leq)$, complete lattice of
subsets of $\State$ ordered by reverse inclusion, ie, $U\leq V\defiff U\supseteq V$. Thus, $\bigsqcup S = \bigcap S$.
\item
$\cL(\State)\defeq(\CS(\State),\leq)$, complete lattice of closed subsets of $\State$ ordered by $\leq$.
\item
$\HS(\State)\defeq\pL(\State^2)\X\pL(\State^2)$, complete lattice of all HS on $\State$.
\item
$\HS_c(\State)\defeq\cL(\State^2)\X\cL(\State^2)$, complete lattice of all closed HS on $\State$.
\end{itemize}
Given a complete lattice $X$ and $x:|X|$, the complete lattice $X\uparrow x$ is the set $\{x'|x\leq_X x'\}$ ordered by $\leq_X$.
When $X$ is one of the complete lattices from Def~\ref{def-lattices}, instead of $X\uparrow x$, we write $x$ in place of $\State$, eg, $\pL(S)$
stands for $\pL(\State)\uparrow S$ when $S:\PS(\State)$, and $\HS(\hs)$ stands for $\HS(\State)\uparrow\hs$ when $\hs:\HS(\State)$.
\end{definition}
\begin{proposition}\label{thm-isos}
In $\Po$ one has the following isomorphisms:
\[\prod i:I.\pL(\State_i)\cong\pL(\sum i:I.\State_i),\qquad
  \prod i:I.\cL(\State_i)\cong\cL(\sum i:I.\State_i),\]
with isomorphism maps $(A_i|i:I)\mapsto\{(i,a)|i:I\land a:A_i\}$.
Thus, $\HS(\State)\cong\pL(\State')$ and $\HS_c(\State)\cong\cL(\State')$, where
$\State'$ is the topological space $\State^2+\State^2\cong 2\X\State\X\State$.

Given $C:\CS(\State)$, equivalently, a closed subspace of $\State$, (1) below is a commuting diagram in $\Po$ of sup-preserving maps, and
(2) is the commuting diagram of right adjoints to the maps in (1), where $g(U)=U\cap C$, and $f^R(U)=\cl{U}$.
\[\commdiag{
\cL(C)&\rIntoA~f&\pL(C)\\
\uTo~g&(1)&\uTo~g\\
\cL(\State)&\rIntoA~f&\pL(\State)
}\qquad\qquad\commdiag{
\cL(C)&\lTo~{f^R}&\pL(C)\\
\dIntoA~{g^R}&(2)&\dIntoA~{g^R}\\
\cL(\State)&\lTo~{f^R}&\pL(\State)
}\]
\end{proposition}
\begin{proof}
By definition of sum $\State=\sum i:I.\State_i$ of topological spaces, a subset $A$ of $\State$ is closed exactly when each $A_i=\{a|(i,a):A\}$ is a closed subset of $\State_i$.

The set $\CS(\State)$ of closed subsets of a topological space $\State$ is closed wrt intersections computed in $\PS(\State)$, thus the inclusion $f$ preserves sups.  Since $C$ is closed, $g(U)=U\cap C$ is closed when $U$ is closed.
\qed\end{proof}
\begin{definition}\label{def-mon-maps}
Given a Banach space $\State$, we define the following maps in $\Po$:
\begin{enumerate}
\item $t:\Po(\HS(\State),\HS(\Real\X\State))$ is the construction on HS in Def~\ref{def-HT}.
\item $\Ts:\Po(\HS(\State)\X\pL(\State),\pL(\State))$ is given by $\Ts(\hs,S)=\{s'|s\rTo_\hs s'\}$.
\item $\Rf{}:\Po(\HS(\State)\X\pL(\State),\pL(\State))$ is the reachability map in Def~\ref{def-R}.
\item $\Rs{}:\Po(\HS(\State)\X\cL(\State),\cL(\State))$ is the safe reachability map in Def~\ref{def-SR}.
\item $\Sp:\Po(\HS(\State),\pL(\State))$ is given by $\Sp(F,G)=\{s|\exists s'.\rel{F}{s}{s'}\lor\rel{G}{s}{s'}\lor\rel{G}{s'}{s}\}$.
\end{enumerate}
We call $\Ts(\hs):\Po(\pL(\State),\pL(\State))$ \textbf{transition map} of $\hs$ and
$\Sp(\hs):\PS(\State)$ its \textbf{support}.
\end{definition}
The transition map of $\hs$ corresponds to the transition relation and suffices for defining the maps $\Rf{}(\hs)$ and $\Rs{}(\hs)$.
The support of $\hs$ does not include the values in the image of $F$, because they represent velocities, instead of states.
\begin{proposition}\label{thm-mon-maps}
Given a Banach space $\State$, the following hold:
\begin{enumerate}
    \item $t(\cl{\hs})=\cl{t(\hs)}$, for every $\hs:\HS(\State)$.
    \item If $\hs:\HS_c(\State)$ and $\Ts(\hs,C)\subseteq C:\cL(\State)$, then
    $\Ts:\Po(\HS(\hs)\X\pL(C),\pL(C))$,\\
    $\Rf{}:\Po(\HS(\hs)\X\pL(C),\pL(C))$, and
    $\Rs{}:\Po(\HS(\hs)\X\cL(C),\cL(C))$.
    \item $\cl{\Sp(\hs)}=\Rf{}(\hs,\cl{\Sp(\hs)})$, for every $\hs:\HS(\State)$.
    \item $\cl{\Sp(\cl{\hs})}=\cl{\Sp(\hs)}\leq \Sp(\cl{\hs})\leq \Sp(\hs):\pL(\State)$, for every $\hs:\HS(\State)$.
    \item If $\hs$ is compact (therefore, closed), then $\Sp(\hs)$ is compact (therefore, closed).
\end{enumerate}
\end{proposition}
\begin{proof}See \ref{sec-proofs}.\end{proof}
Given an adjunction $\commdiag{Y&\pile{\rTo~f\\\bot\\\lDashto~g}&X}$ in $\Po$ such that $g\circ f=\id_Y$, one can identify $Y$ with its image in $X$. Then, the right adjoint $g$ maps every $x:X$ to its \emph{best approximation} in $Y$, ie, $\forall y:Y.y\leq x\iff y\leq g(x)$.  Clearly a \emph{smaller} $Y$ means that the $g(x)$ are less \emph{accurate}.
In \emph{static analysis} it is customary to take as $Y$ a finite lattice, in order to get decidablility.
In the case of safe reachability we take $Y=\cL(\State)$ in place of $X=\pL(\State)$. Also in this case there is a reduction in cardinality: when $\State$ is a Banach space with the cardinality of the continuum, $\CS(\State)$ has the cardinality of the continuuum, while $\pL(\State)$ has a bigger cardinality.
\begin{definition}\label{def-bFa}
Given $X=(|X|,\leq_X):\Po$ and a subset $F\subseteq|X|$, define:
\begin{itemize}
    \item $\FS{X}{F}$ the smallest subset of $|X|$ containing $F$ and closed wrt sups in $X$;
    \item $\FL{X}{F}$ the poset $\FS{X}{F}$ ordered by $\leq_X$;
    \item $\bFa{F}{x}\defeq\bigsqcup\{y:F|y\leq_X x\}:\FL{X}{F}$ the \textbf{best approximation} of $x:X$.
\end{itemize}
\end{definition}
It is quite easy to establish the following basic properties.
\begin{proposition}\label{thm-bFa}
$\FL{X}{F}$ is a complete lattice. Moreover:
\begin{enumerate}
    \item $F\subseteq\FS{X}{F}=\FS{X}{\FS{X}{F}}$.
    \item $F$ is finite$\implies$$\FS{X}{F}$ is finite.
    \item $F_1\subseteq F_2$$\implies$$\FS{X}{F_1}\subseteq\FS{X}{F_2}$ and $\FS{X}{F_1}=\FS{\FL{X}{F_2}}{F_1}$.
    \item The inclusion $f:\FL{X}{F}\rInto X$ is sup-preserving and $f^R(x)=\bFa{F}{x}$.
    \item $\commdiag{Y&\pile{\rTo~f\\\bot\\\lDashto~g}&X}$ and $F$ image of $f$$\implies$$F=\FS{X}{F}$ and $\bFa{F}{x}=f(g(x))$.
\end{enumerate}
\end{proposition}
\begin{example}\label{ex-bFa}
When $Z$ is the complete lattice $\Po(X,Y)$, there are several subsets $F$ of $|Z|$ satisfying $F=\FS{Z}{F}$ which
consist of monotonic maps that preserve certain types of sups. For instance:
\begin{enumerate}
\item \textbf{strict maps}, which preserve the sup of the empty set, ie, $f(\bot)=\bot$.
\item \textbf{additive maps}, which preserve all sups, ie, $f(\bigsqcup S)=\bigsqcup f(S)$ for $S\subseteq|X|$.
\item Scott \textbf{$D$-continuous maps}, which preserve sups of \textbf{directed subsets}, ie, non-empty $D\subseteq|X|$ satisfying
$\forall x,y:D.\exists z:D.x,y\leq_X z$.
\item Scott \textbf{$\omega$-continuous maps}, which preserve sups of \textbf{$\omega$-chains}, ie, sequences $(x_n|n:\omega)$ such that
$\forall n.x_n\leq x_{n+1}$.
\end{enumerate}

The following implications always hold:
\[\mbox{strict $\lImplies$ additive $\rImplies$ $D$-continuous $\rImplies$ $\omega$-continuous}.\]
Requirements can be combined, eg, by defining the \emph{strict $D$-continuous} maps.
\end{example}
\begin{proposition}\label{thm-bca}
The \textbf{best $D$-continuous approximation} $\bca{f}$ of $f$
satisfies the properties:
$\id_X=\bca{\id_X}$ and
$\bca{(\bca{g}\circ\bca{f})}=\bca{g}\circ\bca{f}\leq\bca{(g\circ f)}$.
\end{proposition}
\begin{proof}
Like any best approximation $\bFa{F}{-}$ on $\Po(X,Y)$, $\bca{-}$ is monotonic and satisfies $\bca{(\bca{f})}=\bca{f}\leq f$.
The $D$-continuous maps form a sub-category of $\Po$, ie, they include identities and are closed wrt composition, thus
$\bca{\id_X}=\id_X$ and $\bca{(\bca{g}\circ\bca{f})}=\bca{g}\circ\bca{f}$.
Finally $\bca{g}\circ\bca{f}\leq\bca{(g\circ f)}$ follows from monotonicity of composition and $\bca{-}$. 
\qed\end{proof}

\subsection{Best Approximations on Continuous Lattices}

Thm~\ref{thm-bca-cl} gives a simple way to compute $\bca{f}$ of $f:\Po(X,Y)$ when $X$ is a \emph{continuous} lattice.
Moreover, all continuous lattices for which we want to compute $\bca{f}$ have a countable \emph{base}, which implies that $D$- and $\omega$-continuous maps over these lattices coincide.
\begin{remark}\label{rmk-edalat}
\cite{Edalat95} advocates the use of Domain Theory for the study of dynamical systems.
In this context, the complete lattice $\cL(\State)$ becomes relevant when $\State$ is a compact metric space.
In fact, under this assumption, $\cL(\State)$ is a continuous lattice with a countable base (see \cite[Prop 3.4]{Edalat95}).\footnote{There are some minor differences between our work and that of~\cite{Edalat95}, namely, in~\cite{Edalat95}, $UX$ is a \emph{hyperspace} on the set of non-empty compact subsets of $X$, while our $\cL(\State)$ is a complete lattice on the closed subsets of $\State$.  However, for compact Hausdorff spaces, the only difference is given by the empty subset.} Thus, one can address computability issues, and the Scott topology coincides with the Upper topology (see \cite[Prop 3.2 and 3.3]{Edalat95}).
\end{remark}

 We recall some basic notions and facts on continuous lattices. For more details, the reader is referred to~\cite{Gierz-ContinuousLattices-2003,AbramskyJung94-DT}.

\begin{definition}
Given two complete lattices $X$ and $Y$, define:
\begin{itemize}
    \item the \textbf{way-below} relation: $x\ll_X x'\defiff$ for any directed $D\subseteq|X|$, if $x'\leq_X\bigsqcup D$, then $\exists d:D.x\leq_X d$.
    \item the subset $\twoheaddownarrow x'\defeq\{x|x\ll_X x'\}$ of elements \textbf{way-below} $x'$.
\end{itemize}
The \textbf{step map} $[x,y]:\Po(X,Y)$ is $[x,y](x')\defeq\mbox{if $(x\ll_X x')$ then $y$ else $\bot_Y$}$.
\end{definition}
\begin{proposition}
For any complete lattices $X$ and $Y$, the following hold:
\begin{enumerate}
\item $x_0'\leq_X x_0\ll_X x_1\leq_X x_1'\implies x_0'\ll_X x_1'$.
\item $x_0\ll_X x_1\implies x_0\leq_X x_1$.
\item $x_0,x_1\ll_X x\implies x_0\sqcup x_1\ll_X x$ and $\bot\ll_X x$.
\item $\twoheaddownarrow x$ is directed.
\end{enumerate}
\end{proposition}

\begin{definition}
A complete lattice $X$ is \textbf{continuous} $\defiff\forall x:X.x=\bigsqcup\twoheaddownarrow x$.
A subset $B\subseteq|X|$ is a \textbf{base} for $X$ $\defiff\forall x:X.$ $B\cap\twoheaddownarrow x$ is directed and $x=\bigsqcup(B\cap\twoheaddownarrow x)$.
\end{definition}
\begin{proposition}
Given a continuous lattice $X$ and a complete lattice $Y$:
\begin{enumerate}
\item $D\subseteq|X|$ is directed and $x\ll_X\bigsqcup D\implies \exists d:D.x\ll_X d$.
\item The step map $[x,y]:\Po(X,Y)$ is $D$-continuous.
\item if $X$ has a countable base, then $f:\Po(X,Y)$ is $\omega$-continuous$\implies$$f$ is $D$-continuous.
\end{enumerate}
\end{proposition}
\begin{proof}\ 
\begin{enumerate}
\item \cite[Thm I-1.9]{Gierz-ContinuousLattices-2003}.
\item \cite[Prop 4.0.1]{AbramskyJung94-DT} and $\{x'|x\ll_X x'\}$ Scott open when $X$ is continuous, by the previous point.
\item \cite[Prop 2.2.14]{AbramskyJung94-DT}.\qed
\end{enumerate}
\end{proof}

The following theorem says that $\bca{f}:\Po(X,Y)$ is the sup of step maps when $X$ is a continuous lattice, and
more specifically, $\bca{f}(x)$ is the sup of $\{f(b)|b\ll_X x\}$.
\begin{theorem}\label{thm-bca-cl}
If $f:\Po(X,Y)$ and $B$ is a base for $X$, then $\bca{f}=\bigsqcup_{b:B}[b,f(b)]$.
\end{theorem}
\begin{proof}
As $X$ is a continuous lattice, the maps $[b,f(b)]$ and $f'\defeq\bigsqcup_{b:B}[b,f(b)]$ are $D$-continuous.
We prove that for any given $D$-continuous map $g\leq f$, we have $g\leq f'\leq f$ .
For $x:X$, the subset $B(x)\defeq B\cap\twoheaddownarrow x$ is directed and $x=\bigsqcup B(x)$. Therefore:
\begin{itemize}
\item $\forall b:B(x).g(b)\leq_Y f(b)\leq_Y f(x)$, since $g\leq f$ and $\forall b:B(x).b\leq_X x$.
\item $g(x)=\bigsqcup_{b:B(x)}g(b)$ by continuity of $g$.
\item $f'(x)=\bigsqcup_{b:B(x)}f(b)$ by definition of $f'$.
\end{itemize}
Thus, one concludes that $g(x)\leq_Y f'(x)\leq_Y f(x)$.
\qed\end{proof}

\section{Robustness}\label{sec-robust}

The safe reachability maps introduced in Def~\ref{def-SR} and~\ref{def-mon-maps} are of the form $A:\Po(\cL(\State_1),\cL(\State_2))$, with $\State_1$ and $\State_2$ metric spaces.
We say that such an $A$ is \emph{robust} at $C$ when \emph{small extensions} to $C$ cause small extensions to $A(C)$.

\begin{definition}[Robustness]\label{def-robust}
Given two metric spaces $(\State_i,d_i)$, $i \in \{ 1, 2\}$, we say that a map $A:\Po(\cL(\State_1),\cL(\State_2))$ is \textbf{robust} $\defiff$
\[\forall C:\CS(\State_1).\forall\epsilon>0.\exists\delta>0.A(C_\delta)\subseteq A(C)_\epsilon,\]
where $S_\delta:\CS(\State_i)$ is the closure of the open subset $B(S,\delta)\defeq\{y|\exists x:S.d_i(x,y)<\delta\}$.
\end{definition}

For discrete systems, robustness is not an issue (when the metric space $\State_1$ is \emph{discrete} every monotonic map is robust),
while robustness of safe reachability $\Rs{\hs}:\Po(\cL(\State),\cL(\State))$ for a HS $\hs$ on $\State$ is very sensitive to the HS, and restricting to purely continuous systems does not help either (see Example~\ref{ex-no-robust}).

The main result of this section is that robustness and $D$-continuity coincide
when the $\State_i$ are compact metric spaces. In this case, an analysis $A$ \emph{becomes} robust when replaced by its best $D$-continuous approximation.
Remarkable instances of $A$, for which the main result applies, are safe reachability maps $\Rs{}:\Po(\HS_c(\hs_0)\X\cL(S_0),\cL(S_0))$, where $\hs_0$ is a compact HS on $\State$ with support $S_0$ and $\HS_c(\hs_0)$ is the complete lattice of closed hybrid systems that \emph{refine} $\hs_0$, ie, $\hs\geq \hs_0:\HS_c(\State)$.

In order to relate different properties of monotonic maps in $\Po(X,Y)$, it is conceptually useful to move to the category $\Top$ of topological spaces, by considering suitable topologies on the underlying sets $|X|$ and $|Y|$.
For a complete lattice $X$, one can define two topologies on $|X|$: Alexandrov topology $\OS_A(X)$, and Scott topology $\OS_S(X)\subseteq \OS_A(X)$.
The partial order $\leq_X$ can be recovered as the \emph{specialization order} of the two topologies. Thus, topologies are \emph{more informative} than partial orders: the Alexandrov continuous maps, ie, the maps in $\Top(\OS_A(X),\OS_A(Y))$, are the monotonic maps from $X$ to $Y$, while the Scott continuous maps $\Top(\OS_S(X),\OS_S(Y))$ are the $D$-continuous maps from $X$ to $Y$.
Furthermore, for robustness one can identify a suitable topology on $\CS(\State)$.

\begin{definition}[Topologies]\label{def-top}
When $\State$ is a topological space and $S:\PS(\State)$, let $\UpC{S}\defeq\{C:\CS(\State)|S\leq C\}$.
We define three topologies on $\CS(\State)$:
\begin{itemize}
    \item $U$ \textbf{Alexandrov open} $\defiff$ $U$ is \emph{upward closed}, ie, $C:U\implies \UpC{C}\subseteq U$.
    \item $U$ \textbf{Scott open} $\defiff$ $U$ is \emph{upward closed} and $(\bigsqcup D):U\implies\exists C:D.C:U$, for any directed subset $D$ of closed subsets.
    \item $U$ open in the \textbf{Upper topology} $\defiff$ $C:U\implies\exists O:\OS(\State).C:\UpC{O}\subseteq U$.
\end{itemize}
Moreover, when $(\State,d)$ is a metric space, we define a fourth topology on $\CS(\State)$:
\begin{itemize}
    \item $U$ open in the \textbf{Robust topology} $\defiff$ $C:U\implies\exists\delta>0.\UpC{B(C,\delta)}\subseteq U$.
\end{itemize}
\end{definition}

\begin{lemma}\label{thm-top}
In a metric space $(\State,d)$, the following implications hold:
\begin{enumerate}
\item $U$ Scott open $\implies$
\item $U$ open for the Robust topology $\implies$
\item $U$ open for the Upper topology $\implies$
\item $U$ Alexandrov open.
\end{enumerate}
\end{lemma}
\begin{proof}
We exploit the following facts, which are valid in any metric space:
\begin{itemize}
\item[(F1)] $B(B(S,\delta),\delta')\subseteq B(S,\delta+\delta')$.
\item[(F2)] $S\subseteq \cl{S}=\cap_{\delta>0} B(S,\delta)$.
\item[(F3)] $B(S,\delta)\subseteq S_\delta\defeq\cl{B(S,\delta)}\subseteq B(S,\delta')$ whenever $\delta<\delta'$.
\end{itemize}
Now we prove the implications:
\begin{itemize}
\item (1)$\implies$(2).
If $U$ is Scott open and $C:U$, let $O_n\defeq B(C,2^{-n})$,
we prove that $\UpC{O_n}\subseteq U$ for some $n$.

Let $C_n\defeq\cl{O_n}$. Then, $C_n\supseteq O_n\supseteq C_{n+1}$, and
$C=\cl{C}=(\cap_n C_n):U$. Thus, $C_n:U$, for some $n$.

Since $O_n\subseteq C_n:U$ and $U$ is upward closed, then $\UpC{O_n}\subseteq U$.

\item (2)$\implies$(3).
If $U$ is open in the Robust topology and $C:U$, by definition, there exists a $\delta>0$ for which $\UpC{B(C,\delta)}\subseteq U$. Then, $B(C,\delta)$ is the open subset that witnesses that $U$ is open in the Upper topology.

\item (3)$\implies$(4).
If $U$ is open in the Upper topology, then $U$ is the union of upward closed subsets of the form $\UpC{O}$,
therefore it is Alexandrov open.\qed
\end{itemize}
\end{proof}
As anticipated, the Robust topology captures the robustness property:
\begin{theorem}\label{thm-robust-equal}
Given two metric spaces $(\State_i,d_i)$, and a map $A:\Set(\CS(\State_1),\CS(\State_2))$, the following properties are equivalent:
\begin{enumerate}
\item $A$ is monotonic and robust in the sense of Def~\ref{def-robust}.
\item $A$ is continuous wrt the Robust topologies.
\end{enumerate}
\end{theorem}
\begin{proof}See \ref{sec-proofs}.\end{proof}

\begin{lemma}\label{thm-top-equal}
We have the following inclusions among topologies on $\CS(\State)$:
\begin{enumerate}
\item The Upper topology is included in the Scott topology, when $\State$ is compact.
\item The Alexandrov topology is included in the Robust topology, when $\State$ is a discrete metric space (eg, $d(x,y)=1$ when $x\neq y$).
\end{enumerate}
\end{lemma}

\begin{proof}
When $\State$ is a compact topological space, any closed subset of $\State$ is compact.
To prove that the Upper topology is included in the Scott topology, it suffices to show that for any open subset $O$ of $\State$, the subset $\UpC{O}$ of $\CS(\State)$ is Scott open.
Clearly $\UpC{O}$ is upward closed. Now, let $C$ be the complement of $O$. If $\{K_i|i:I\}$ is an $I$-indexed directed set of closed subsets such that $K=(\cap_{i:I} K_i)\subseteq O$, then $K'_i=K_i\cap C$ is another $I$-indexed directed subset $D'$ of closed subsets whose intersection is $\emptyset$. Since compact subsets (therefore, also the closed subsets) have the \emph{finite intersection property} and $D'$ is directed, some $K'_i$ must be $\emptyset$, or equivalently, $K_i:\UpC{O}$.

If $(\State,d)$ is discrete, then $S=\cl{S}=B(S,1)$, as all subsets are both open and closed.  Therefore, when $U$ is Alexandrov open, ie, $U$ is the union of all $\UpC{C}$ with $C:U$, it is also open in the Robust topology, because $\UpC{C}=\UpC{B(C,1)}$.
\qed\end{proof}
\begin{remark}
The Robust topology depends on the metric $d$, while the other topologies on $\CS(\State)$ are definable for any topological space $\State$.
There are metric spaces $(\State,d)$ where the four topologies on $\CS(\State)$ differ.
For instance, in the space $\State=\{x|x>0\}$ of positive reals with the usual metric $d(x,y)=|y-x|$,
let $x_n\defeq 2^{-n}$ and $\delta_n\defeq x_{n+1}$. Then,
$C=\{x_n|n:\omega\}$ is a closed subset of $\State$, $O=\cup_n B(x_n,\delta_n)$ is an open subset of $\State$, and
the following counter-examples show that the four topologies differ:
\begin{itemize}
\item $\UpC{C}$ is Alexandrov open, but it is not open in the other topologies.
\item $\UpC{O}$ is open in the Upper topology, but it is not open in the Robust topology, because $C:\UpC{O}$ but there is no $\delta>0$ such that $B(C,\delta)\subseteq O$, because $(0,\delta)\subseteq B(C,\delta)$ but $(0,\delta)\not\subseteq O$.
\item $\UpC\emptyset=\{\emptyset\}$ is open in the Robust topology, because $B(\emptyset,\delta)=\emptyset$, but it is not Scott open, because $C_n=\{x|x\geq 2^{n}\}$ is an $\omega$-chain of closed non-empty subsets of $\State$ whose intersection is $\emptyset$.
\end{itemize}
In finite dimensional Banach spaces like $\Real^n$, the compact subsets are exactly the closed subsets $C$ that are \emph{bounded}, ie, $C\subset B(0,\delta)$ for some $\delta>0$.  On the contrary, in infinite dimensional Banach spaces (under the strong topology) they form a proper subset of the closed bounded ones. For instance, neither the closure nor the boundary of a ball $B(0,\delta)$ is compact.
\end{remark}

\begin{theorem}\label{thm-robust-main}
Given a map $A:\Po(\cL(\State_1),\cL(\State_2))$, with  $\State_1$ and $\State_2$ metric spaces:
\begin{enumerate}
\item If $\State_1$ and $\State_2$ are compact, then, $A$ is robust $\iff$ $A$ is Scott continuous.
\item If $\State_1$ and $\State_2$ are discrete, then, $A$ is robust.
\end{enumerate}
\end{theorem}

\begin{proof}
We use Lemma~\ref{thm-top}~and~\ref{thm-top-equal}:
\begin{enumerate}
\item If $\State_i$ is a compact metric space, the Upper and Scott topologies on $\CS(\State_i)$ coincide.  The result follows as the Robust topology is between the two.
\item If $\State_i$ is a discrete metric space, the Alexandrov and Robust topologies on $\CS(\State_i)$ coincide. The result follows as every the monotonic map is Alexandrov continuous.\qed
\end{enumerate}
\end{proof}
\begin{remark}
Theorem~\ref{thm-robust-main} says that for discrete metric spaces robustness is not an interesting notion, because every monotonic map $A:\Po(\cL(\State_1),\cL(\State_2))$ is robust, and moreover $\cL(\State_i)=\pL(\State_i)$.
When the metric spaces are compact,  there is a best robust approximation of $A$ given by $\bca{A}$.
In all other cases robustness is an interesting property, but there is no \emph{best way} to approximate $A$ with a robust map, because the subset of robust maps is not closed wrt arbitrary sups computed in $\Po(\cL(\State_1),\cL(\State_2))$.
\end{remark}

\begin{example}\label{ex-no-robust}
We give an example of a compact HS $\hs$ with support $C$, whose safe reachability map $\Rs{\hs}:\Po(\cL(C),\cL(C))$ is not robust.

Let $\hs$ be $\hs_E$ of Example~\ref{ex-expand}---which can only flow---and let $C=[0,M]$.
Since $C$ is compact, robustness is equivalent to Scott continuity.
Consider the $\omega$-chain $(I_n|n)$ with $I_n=[0,M/2^n]$, whose sup is the singleton $\{0\}$. Then:

\[\Rs{\hs}(I_n)=[0,M],\qquad \Rs{\hs}(\{0\})=\{0\}.\]

Thus, $\Rs{\hs}$  does not preserve the sup of the $\omega$-chain $(I_n|n)$.
Consider its best Scott continuous approximation $\bca{\Rs{\hs}}$ (see Thm~\ref{thm-bca}), which is given by:
\begin{itemize}
\item $\bca{\Rs{\hs}}(\emptyset)=\emptyset$.
\item $\bca{\Rs{\hs}}(U)=[m,M]$ with $m=\min(U)$, when $U \neq\emptyset$ is compact.
\end{itemize}
Then, $\Rs{\hs}(\{0\})=\{0\}\subset[0,M] =\bca{\Rs{\hs}}(\{0\})$.
\end{example}

\definecolor{zenoGreen}{rgb}{0.1,0.8,0.1}
\newcommand\colorIC{blue}
\newcommand\colorZeno{zenoGreen}
\newcommand\colorTraj{blue}
\newcommand\colorReachSet{blue}
\newcommand\colorSafereachSet{zenoGreen}
\newcommand\colorUnreachSet{red}
\newcommand\widthTraj{thick}
\newcommand\widthIC{1.5pt}

\pgfplotsset{compat = 1.8}

\section{Figures and Examples}\label{sec-figs}
We go through the hybrid systems introduced in Sec~\ref{sec-HS} and for each of them we compare the sets computed by different analyses.  More precisely, given a HS $\hs$ on $\State$ and a state $s_0$ in the support $\Sp(\hs)$ of $\hs$, take as set of initial states $I=\{s_0\}$, then define four subsets $S_f\subseteq S_s\subseteq S_r\subseteq S_R$ of $\State$:
\begin{itemize}
\item $S_f\defeq\Rf{\hs}(I)$ set of states reachable in finitely many transition.
\item $S_s\defeq\Rs{\hs}(I)$ safe approximation of the set of states reachable in finite time.  One should use $\cl{I}$ in place of $I$, but a singleton is already closed.
\item[] To define $S_r$ one must restrict $\Rs{\hs}$ to a map in $\Po(\cL(S_0),\cL(S_0))$, where $S_0$ is a \emph{sufficiently large} compact subset of $\State$.  When $\cl{\hs}$ is compact, the canonical choice for $S_0$ is $\Sp(\cl{\hs})$.
\item $S_r\defeq\bca{\Rs{\hs}}(I)$ approximation of $S_s$ robust wrt perturbations to $I$.
\item[] To define $S_R$ one must restrict $\Rs{}$ to a map in $\Po(\HS_c(\hs_0)\X\cL(S_0),\cL(S_0))$, where $\hs_0$ is a \emph{sufficiently large} compact HS on $\State$ and $S_0=\Sp(\hs_0)$. There is no canonical choice. $\hs_0$ must capture the \emph{allowed} perturbations to $\hs$, thus $\hs_0\leq\cl{\hs}$ in $\HS_c(\State)$, where $\cl{\hs}$ is the closure of $\hs$.
\item $S_R\defeq\bca{\Rs{}}(\cl{\hs},I)$ approximation of $S'_r\defeq\bca{\Rs{\cl{\hs}}}(I)$ robust wrt perturbations to $\hs$ \& $I$ \emph{allowed} by $\hs_0$.
\end{itemize}
In Figures we adopt the following color coding for states and trajectories:
\begin{itemize}
\item a bullet \raisebox{0.1em}{\begin{tikzpicture}
 \filldraw [\colorIC] (0,0) circle (\widthIC);
\end{tikzpicture}} indicates the initial state $s_0$
\item blue - \textcolor{\colorReachSet}{$S_f$} and the \textcolor{\colorTraj}{part of a trajectory} reachable in finitely many transitions
\item green - \textcolor{\colorSafereachSet}{$S_s-S_f$} and the \textcolor{\colorZeno}{rest of a trajectory} not reachable in finitely many transitions, like Zeno points and beyond
\item red - \textcolor{\colorUnreachSet}{$S_r-S_s$}, there is no analogue for a trajectory starting from $s_0$.
\end{itemize}

\pgfarrowsdeclare{bracket}{bracket}
{
\arrowsize=0.2pt
\advance\arrowsize by .5\pgflinewidth
\pgfarrowsleftextend{-4\arrowsize-.5\pgflinewidth}
\pgfarrowsrightextend{.5\pgflinewidth}
}
{
\arrowsize=0.2pt
\advance\arrowsize by .5\pgflinewidth
\pgfsetdash{}{0pt} 
\pgfsetroundjoin 
\pgfsetroundcap 
\pgfpathmoveto{\pgfpoint{-4\arrowsize}{4\arrowsize}}
\pgfpathlineto{\pgfpoint{0\arrowsize}{4\arrowsize}}
\pgfpathlineto{\pgfpoint{0\arrowsize}{-4\arrowsize}}
\pgfpathlineto{\pgfpoint{-4\arrowsize}{-4\arrowsize}}
\pgfusepathqstroke
}

\pgfarrowsdeclare{parenthesis}{parenthesis}
{
\arrowsize=0.2pt
\advance\arrowsize by .5\pgflinewidth
\pgfarrowsleftextend{-4\arrowsize-.5\pgflinewidth}
\pgfarrowsrightextend{.5\pgflinewidth}
}
{
\arrowsize=0.2pt
\advance\arrowsize by .5\pgflinewidth
\pgfsetdash{}{0pt} 
\pgfsetroundjoin 
\pgfsetroundcap 
\pgfpathmoveto{\pgfpoint{-4\arrowsize}{4\arrowsize}}
\pgfpatharc{90}{-90}{4\arrowsize}
\pgfusepathqstroke
}

\subsection{Expand}\label{sec-ex-expland}

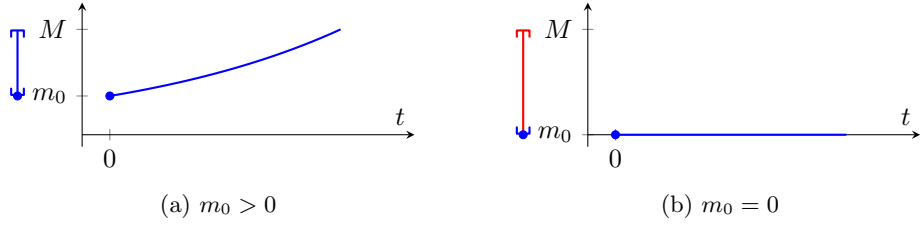
\begin{figure}[t]
\centering

\begin{subfigure}[b]{0.45\textwidth}

\begin{tikzpicture}
  \begin{axis} 
      [ width = 6cm
      , height = 3.5cm
      , xlabel = $t$
      , y label style={
          at={(axis description cs:0.435,1)},rotate=-90,anchor=north
        }
      , axis x line = middle
      , axis y line = left
      , xmax = 1.2
      , ymax = 3.1
      , ymin = 0.0
      , xtick = {0}
      , ytick = {1,2.7182818285}
      , yticklabels={$m_0$,$M$}
      , enlarge x limits=true
      , enlarge y limits=true
      , clip=false
      ]
      \addplot [\colorTraj,\widthTraj,domain=0:1] { exp(x) };
      \filldraw[\colorIC] (axis cs: 0,1) circle (\widthIC);
      \draw [thick, bracket-bracket, \colorReachSet]
            (axis cs: -0.4,1) -- 
            (axis cs: -0.4,2.7182818285);
      \filldraw [\colorIC]
            (axis cs: -0.4,1) circle (\widthIC);
  \end{axis}
\end{tikzpicture}
\caption{$m_0 > 0$}

\end{subfigure}
\hfill
\begin{subfigure}[b]{0.45\textwidth}
\centering
\begin{tikzpicture}
  \begin{axis} 
      [ width = 6cm
      , height = 3.5cm
      , xlabel = $t$
      , y label style={
          at={(axis description cs:0.435,1)},rotate=-90,anchor=north
        }
      , axis x line = middle
      , axis y line = left
      , xmax = 1.2
      , ymax = 3.1
      , ymin = 0.0
      , xtick = {0}
      , ytick = {0,2.7182818285}
      , yticklabels={$m_0$,$M$}
      , enlarge x limits=true
      , enlarge y limits=true
      , clip = false
      ]
      \addplot [\colorTraj,\widthTraj,domain=0:1] { 0 };
      \filldraw[\colorIC] (axis cs: 0,0) circle (\widthIC);
      \draw [thick, bracket-, \colorReachSet]
            (axis cs: -0.4,0) -- 
            (axis cs: -0.4,0.00001);
      \draw [thick, -bracket, red]
            (axis cs: -0.4,0) -- 
            (axis cs: -0.4,2.7182818285);
      \filldraw [\colorIC]
            (axis cs: -0.4,0) circle (\widthIC);
  \end{axis}
\end{tikzpicture}
\caption{$m_0 = 0$}

\end{subfigure}

\caption{Trajectories and reachable states of $\hs_E$ (Expand).
In each case the set of states reachable from $I=\{m_0\}$ is on the left of the trajectory starting from $m_0$.}
\label{fig:expand-trajectories}
\end{figure}
$\hs_E$ of Example~\ref{ex-expand} is a compact \emph{deterministic} HS on $\Real$, whose behavior is depicted in Fig~\ref{fig:expand-trajectories}.  The canonical choice for $S_0$ is the interval $[0,M]$.
For $\hs_0$ we take $F_0=S_0\X[-M,M]$ and $G_0=\emptyset$, whose support is still $S_0$.
\begin{itemize}
\item $S_f=S_s=S_r=S_R=[m_0,M]$ when $0<m_0\leq M$
\item $S_f=S_s=[0]\subset [0,M]=S_r=S_R$ when $0=m_0$.
\end{itemize}
We now explain why making the set of reachable states robust wrt perturbations to $\hs$ does not make a difference
in the case $0<m_0\leq M$.
To approximate $S_R$ we take a small $\delta>0$ and define $I_\delta\ll I$ in $\cL(S_0)$ and $\hs_\delta\ll\hs$ in $\HS_c(\hs_0)$.

Let $I_\delta=[m_0-\delta,M]$ and $F_\delta=\{(m,\dot{m})|0\leq m\leq M\land m-\delta\leq\dot{m}\leq M\}$.
By taking $2*\delta<m_0$ we can ensure that $F_\delta(m)\subset (0,M]$ for any $m:I_\delta$, therefore
$S_R=\bca{\Rs{}}(\cl{\hs},I)\subseteq\Rs{}(\hs_\delta,I_\delta)=[m_0-\delta,M]\to S_r$ when $\delta\to 0$.

\subsection{Decay}
$\hs_D$ of Example~\ref{ex-decay} is a deterministic HS on $\Real$, whose behavior is depicted in Fig~\ref{fig:decay-trajectories}.
Its closure is compact but it is no longer deterministic.  The canonical choice for $S_0$ is the interval $[0,M]$.
For $\hs_0$ we take $F_0=S_0\X[-M,M]$ and $G_0=S_0\times S_0$, whose support is still $S_0$.

\begin{figure}[th]
\centering

\begin{subfigure}[b]{0.45\textwidth}

\begin{tikzpicture}
  \begin{axis} 
      [ width = 6cm
      , height = 3.5cm
      , xlabel = $t$
      , y label style={
          at={(axis description cs:0.435,1)},rotate=-90,anchor=north
        }
      , axis x line = middle
      , axis y line = left
      , xmax = 1.2
      , ymax = 2
      , ymin = 0.0
      , xtick = {0}
      , ytick = {1,1.75}
      , yticklabels={$m_0$,$M$}
      , enlarge x limits=true
      , enlarge y limits=true
      , clip = false
      ]
      \addplot [\colorTraj,\widthTraj,domain=0:1] { exp(-x) };
      \filldraw[\colorIC] (axis cs: 0,1) circle (\widthIC);
      \draw [thick, parenthesis-bracket, \colorReachSet]
            (axis cs: -0.4,0) -- 
            (axis cs: -0.4,1);
      \draw [thick, -bracket, \colorSafereachSet]
            (axis cs: -0.4,1) -- 
            (axis cs: -0.4,1.75);
      \filldraw [\colorIC]
            (axis cs: -0.4,1) circle (\widthIC);
      \filldraw [\colorSafereachSet]
            (axis cs: -0.4,0) circle (\widthIC);
  \end{axis}
\end{tikzpicture}
\caption{$m_0 > 0$}

\end{subfigure}
\hfill
\begin{subfigure}[b]{0.45\textwidth}
\centering
\begin{tikzpicture}
  \begin{axis} 
      [ width = 6cm
      , height = 3.5cm
      , xlabel = $t$
      , y label style={
          at={(axis description cs:0.435,1)},rotate=-90,anchor=north
        }
      , axis x line = middle
      , axis y line = left
      , xmax = 1.2
      , ymax = 2
      , ymin = 0.0
      , xtick = {0}
      , ytick = {0,1.75}
      , yticklabels={$m_0$,$M$}
      , enlarge x limits=true
      , enlarge y limits=true
      , clip = false
      ]
      \addplot [\colorTraj,\widthTraj,domain=0:1] { 1.75*exp(-x) };
      \filldraw[\colorIC] (axis cs: 0,0) circle 
        (\widthIC);
      \filldraw[draw=black,fill=white] (axis cs: 0,1.75) circle 
        (\widthIC);
      \draw [thick, bracket-bracket, \colorReachSet]
            (axis cs: -0.4,0) -- 
            (axis cs: -0.4,1.75);
      \filldraw [\colorIC]
            (axis cs: -0.4,0) circle (\widthIC);
  \end{axis}
\end{tikzpicture}
\caption{$m_0 = 0$}

\end{subfigure}

\caption{Trajectories and reachable states of $\hs_D$ (Decay).
In each case the set of states reachable from $I=\{m_0\}$ is on the left of the trajectory starting from $m_0$.}
\label{fig:decay-trajectories}
\end{figure}
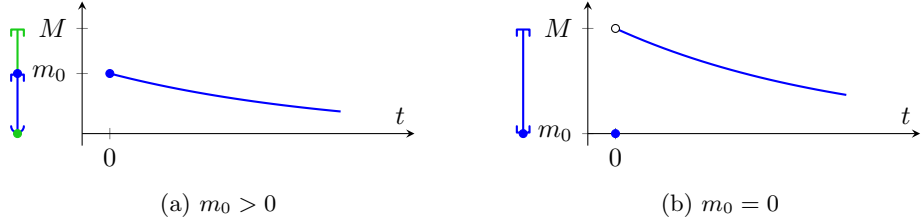
\begin{itemize}
\item $S_f=(0,m_0]\subset[0,M]=S_s=S_r=S_R$ when $0<m_0\leq M$
\item $S_f=S_s=S_r=S_R=[0,M]$ when $0=m_0$.
\end{itemize}
$S_s=S_r=S_R=S_0$, because $S_s=S_0$ and these subsets cannot be bigger than the support of $\hs_0$.
This result does not change, when $\hs_0$ is replaced with a HS with a bigger support, but the proof is not as simple
(see  Sec~\ref{sec-ex-expland}).

\subsection{Bouncing Ball}
$\hs_B$ of Example~\ref{ex-BB} is a deterministic HS on $\Real^2$, its behavior depends on the coefficient of restitution $b$
(see Fig~\ref{fig:bb-trajectories}).  The closure of $\hs_B$ is not compact and its support is the closed subset
$\{(h,v)|0\geq h\}$.  However, compactness is irrelevant to define and compare $S_f$ and $S_s$.
Let $s_0=(0,v_0)$ with $0<v_0<V$ and $S(u)\defeq\{(h,v)|0\geq h\land E(h,v)=E(0,u)\}$ be the set of states whose energy $E(h,v)=h+\frac{v^2}{2}$ is exactly $E(0,u)$, then the sets $S_f$ and $S_s$ are (see Fig~\ref{fig:bb-states}):

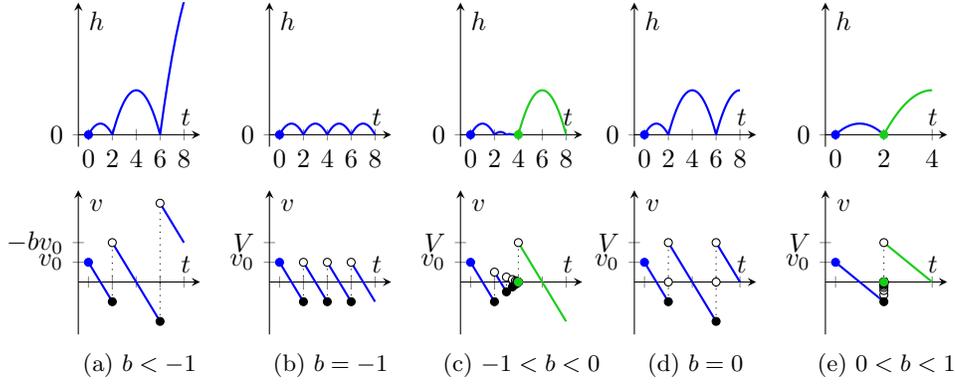
\begin{figure}[t]
\centering
\newcommand\widthSubFigBB{0.181\textwidth}
\newcommand\widthSubFigBBAxis{3.2cm}
\newcommand\ymaxBBh{54}
\newcommand\yminBBh{0.0}
\newcommand\ymaxBBv{40.5}
\newcommand\yminBBv{-20.5}

\begin{subfigure}[b]{\widthSubFigBB}
\centering
\begin{tikzpicture}
  [ trim axis left 
  ]
  \begin{axis} 
      [ width = \widthSubFigBBAxis
      , height = 3.5cm
      , xlabel = $t$
      , ylabel = $h$
      , y label style={
          at={(axis description cs:0.15,1)},rotate=-90,anchor=north
        }
      , axis x line = middle
      , axis y line = left
      , xmax = 8.5
      , ymax = 20.5
      , ymin = 0.0
      , ymax = \ymaxBBh
      , xtick = {0,2,4,6,8}
      , ytick = {0}
      , yticklabels={0}
      , enlarge x limits=true
      , enlarge y limits=true
      ]
        \foreach \zl/\zr in 
            { 0.0/2.0
            , 2.0/6.0
            , 6.0/14.0
            } {
          \edef\temp{\noexpand
            \addplot [\colorTraj,\widthTraj,domain=\zl:\zr] 
              { -5*(x-\zl)*(x-\zr) };
          }
          \temp
        }
        \filldraw[\colorIC] 
          (axis cs: 0,0) circle (\widthIC);
  \end{axis}
  \begin{axis} 
      [ width = \widthSubFigBBAxis
      , height = 3.5cm
      , xlabel = $t$
      , ylabel = $v$
      , y label style={
          at={(axis description cs:0.15,1)},rotate=-90,anchor=north
        }
      , yshift = -2.5cm
      , axis x line = middle
      , axis y line = left
      , xmax = 8.5
      , ymax = \ymaxBBv
      , ymin = \yminBBv
      , xtick = {0,2,4,6,8}
      , xticklabels={}
      , ytick = {10,20}
      , yticklabels={$v_0$,$-bv_0$}
      , enlarge x limits=true
      , enlarge y limits=true
      ]
        \addplot [\colorTraj,\widthTraj,domain=0:2] 
          { 10-10*x };
        \draw[dotted] 
          (axis cs: 2,-10) -- (axis cs: 2,20);
        \filldraw[black] 
          ( axis cs: 2
          , -10
          ) circle (\widthIC);
        \addplot [\colorTraj,\widthTraj,domain=2:6] 
          { 40-10*x };
        \filldraw[draw=black,fill=white] 
          ( axis cs: 2
          , 20
          ) circle (\widthIC);
        \draw[dotted] 
          (axis cs: 6,-20) -- (axis cs: 6,40);
        \filldraw[black] 
          ( axis cs: 6
          , -20
          ) circle (\widthIC);
        \addplot [\colorTraj,\widthTraj,domain=6:8] 
          { 100-10*x };
        \filldraw[draw=black,fill=white] 
          ( axis cs: 6
          , 40
          ) circle (\widthIC);

        \filldraw[\colorIC] 
          (axis cs: 0,10) circle (\widthIC);
  \end{axis}
\end{tikzpicture}
\caption{$b < -1$}

\end{subfigure}
\hfill
\begin{subfigure}[b]{\widthSubFigBB}
\centering
\begin{tikzpicture}
  [ trim axis left 
  ]
  \begin{axis} 
      [ width = \widthSubFigBBAxis
      , height = 3.5cm
      , xlabel = $t$
      , ylabel = $h$
      , y label style={
          at={(axis description cs:0.15,1)},rotate=-90,anchor=north
        }
      , axis x line = middle
      , axis y line = left
      , xmax = 8.5
      , ymax = 20.5
      , ymin = 0.0
      , ymax = \ymaxBBh
      , xtick = {0,2,4,6,8}
      , ytick = {0}
      , yticklabels={0}
      , enlarge x limits=true
      , enlarge y limits=true
      ]
        \foreach \zl/\zr in 
            { 0.0/2.0
            , 2.0/4.0
            , 4.0/6.0
            , 6.0/8.0
            } {
          \edef\temp{\noexpand
            \addplot [\colorTraj,\widthTraj,domain=\zl:\zr] 
              { -5*(x-\zl)*(x-\zr) };
          }
          \temp
        }
        \filldraw[\colorIC] 
          (axis cs: 0,0) circle (\widthIC);
  \end{axis}
  \begin{axis} 
      [ width = \widthSubFigBBAxis
      , height = 3.5cm
      , xlabel = $t$
      , ylabel = $v$
      , y label style={
          at={(axis description cs:0.15,1)},rotate=-90,anchor=north
        }
      , yshift = -2.5cm
      , axis x line = middle
      , axis y line = left
      , xmax = 8.5
      , ymax = \ymaxBBv
      , ymin = \yminBBv
      , xtick = {0,2,4,6,8}
      , xticklabels={}
      , ytick = {10,20}
      , yticklabels={$v_0$,$V$}
      , enlarge x limits=true
      , enlarge y limits=true
      ]
        \addplot 
          [ \colorTraj
          , \widthTraj
          , domain=0:2
          ] { 10-10*x };
        \foreach \zi/\zl/\zr
          [remember = \zi as \zil] in 
            { -10.0/2.0/4.0
            , -10.0/4.0/6.0
            , -10.0/6.0/8.0
            } {
          \edef\tempTraj{\noexpand
            \addplot 
              [ \colorTraj
              , \widthTraj
              , domain=\zl:\zr
              ] { 10-10*(x-\zl) };
          }
          \edef\tempReset{\noexpand
            \draw[dotted] 
              (axis cs: \zl,\zi) -- (axis cs: \zl,-\zi);
          }
          \edef\tempSource{\noexpand
            \filldraw[black] 
              ( axis cs: \zl
              , \zi
              ) circle (\widthIC);
          }
          \edef\tempTarget{\noexpand
            \filldraw[draw=black,fill=white] 
              ( axis cs: \zl
              , -\zi
              ) circle (\widthIC);
          }
          \tempTraj;
          \tempReset;
          \tempSource;
          \tempTarget;
        }
        \filldraw[\colorIC] 
          (axis cs: 0,10) circle (\widthIC);
  \end{axis}
\end{tikzpicture}
\caption{$b = -1$}

\end{subfigure}
\hfill
\begin{subfigure}[b]{\widthSubFigBB}
\centering

\begin{tikzpicture}
  [ trim axis left 
  ]
  \begin{axis} 
      [ width = \widthSubFigBBAxis
      , height = 3.5cm
      , xlabel = $t$
      , ylabel = $h$
      , y label style={
          at={(axis description cs:0.15,1)},rotate=-90,anchor=north
        }
      , axis x line = middle
      , axis y line = left
      , xmax = 8.5
      , ymin = 0.0
      , ymax = \ymaxBBh
      , xtick = {0,2,4,6,8}
      , ytick = {0}
      , yticklabels={0}
      , enlarge x limits=true
      , enlarge y limits=true
      ]
        \foreach \zl/\zr in 
            { 0.0/2.0
            , 2.0/3.0
            , 3.0/3.5
            , 3.5/3.75
            , 3.75/3.875
            , 3.875/3.9375
            , 3.9375/3.96875
            , 3.96875/3.984375
            , 3.984375/3.9921875
            , 3.9921875/3.99609375
            , 3.99609375/3.998046875
            , 3.998046875/3.9990234375
            } {
          \edef\temp{\noexpand
            \addplot [\colorTraj,\widthTraj,domain=\zl:\zr] 
              { -5*(x-\zl)*(x-\zr) };
          }
          \temp
        }
        \addplot [\colorZeno,\widthTraj,domain=4:8] 
          { -5*(x-4)*(x-8) };
        \filldraw[\colorIC] 
          (axis cs: 0,0) circle (\widthIC);
        \filldraw[\colorZeno] 
          (axis cs: 4,0) circle (\widthIC);
  \end{axis}
  \begin{axis} 
      [ width = \widthSubFigBBAxis
      , height = 3.5cm
      , xlabel = $t$
      , ylabel = $v$
      , y label style={
          at={(axis description cs:0.15,1)},rotate=-90,anchor=north
        }
      , yshift = -2.5cm
      , axis x line = middle
      , axis y line = left
      , xmax = 8.5
      , ymax = \ymaxBBv
      , ymin = \yminBBv
      , xtick = {0,2,4,6,8}
      , xticklabels={}
      , ytick = {10,20}
      , yticklabels={$v_0$,$V$}
      , enlarge x limits=true
      , enlarge y limits=true
      ]
        \addplot [\colorTraj,\widthTraj,domain=0:2] { 10-10*x };
        \foreach \zi/\zl/\zr
          [remember = \zi as \zil] in 
            { -10.0/2.0/3.0
            , -5.0/3.0/3.5
            , -2.5/3.5/3.75
            , -1.25/3.75/3.875
            , -0.625/3.875/3.9375
            , -0.3125/3.9375/3.96875
            , -0.15625/3.96875/3.984375
            , -7.8125e-2/3.984375/3.9921875
            , -3.90625e-2/3.9921875/3.99609375
            , -1.953125e-2/3.99609375/3.998046875
            , -9.765625e-3/3.998046875/3.9990234375
            } {
          \edef\tempTraj{\noexpand
            \addplot 
              [ \colorTraj
              , \widthTraj
              , domain=\zl:\zr
              ] { -10*(x-((\zl+\zr)/2)) };
          }
          \edef\tempReset{\noexpand
            \draw[dotted] 
              (axis cs: \zl,\zi) -- (axis cs: \zl,-\zi/2);
          }
          \edef\tempSource{\noexpand
            \filldraw[black] 
              ( axis cs: \zl
              , \zi
              ) circle (\widthIC);
          }
          \edef\tempTarget{\noexpand
            \filldraw[draw=black,fill=white] 
              ( axis cs: \zl
              , -\zi/2
              ) circle (\widthIC);
          }
          \tempTraj;
          \tempReset;
          \tempSource;
          \tempTarget;
        }
        \addplot [\colorZeno,\widthTraj,domain=4:8] 
          { -10*(x-6) };
        \draw[dotted] 
          (axis cs: 4,0) -- (axis cs: 4,20);
        \filldraw[black] 
          (axis cs: 4, 0) circle (\widthIC);
        \filldraw[draw=black,fill=white] 
          (axis cs: 4, 20) circle (\widthIC);
        \filldraw[\colorIC] 
          (axis cs: 0,10) circle (\widthIC);
        \filldraw[\colorZeno] 
          (axis cs: 4,0) circle (\widthIC);
  \end{axis}
\end{tikzpicture}
\caption{$-1 < b < 0$}

\end{subfigure}
\begin{subfigure}[b]{\widthSubFigBB}
\centering

\begin{tikzpicture}
  [ trim axis left 
  ]
  \begin{axis} 
      [ width = \widthSubFigBBAxis
      , height = 3.5cm
      , xlabel = $t$
      , ylabel = $h$
      , y label style={
          at={(axis description cs:0.15,1)},rotate=-90,anchor=north
        }
      , axis x line = middle
      , axis y line = left
      , xmax = 8.5
      , ymin = 0.0
      , ymax = \ymaxBBh
      , xtick = {0,2,4,6,8}
      , ytick = {0}
      , yticklabels={0}
      , enlarge x limits=true
      , enlarge y limits=true
      ]
        \addplot [\colorTraj,\widthTraj,domain=0:2] 
          { -5*x*(x-2) };
        \addplot [\colorTraj,\widthTraj,domain=2:6] 
          { -5*(x-2)*(x-6) };
        \addplot [\colorTraj,\widthTraj,domain=6:8] 
          { -5*(x-6)*(x-10) };
        \filldraw[\colorIC] 
          (axis cs: 0,0) circle (\widthIC);
  \end{axis}
  \begin{axis} 
      [ width = \widthSubFigBBAxis
      , height = 3.5cm
      , xlabel = $t$
      , ylabel = $v$
      , y label style={
          at={(axis description cs:0.15,1)},rotate=-90,anchor=north
        }
      , yshift = -2.5cm
      , axis x line = middle
      , axis y line = left
      , xmax = 8.5
      , ymax = \ymaxBBv
      , ymin = \yminBBv
      , xtick = {0,2,4,6,8}
      , xticklabels={}
      , ytick = {10,20}
      , yticklabels={$v_0$,$V$}
      , enlarge x limits=true
      , enlarge y limits=true
      ]
        \addplot [\colorTraj,\widthTraj,domain=0:2] 
          { 10-10*x };
        \addplot [\colorTraj,\widthTraj,domain=2:6] 
          { 40-10*x };
        \draw[dotted] 
          (axis cs: 2,-10) -- (axis cs: 2,20);
        \filldraw[black] 
          (axis cs: 2, -10) circle (\widthIC);
        \filldraw[draw=black,fill=white] 
          (axis cs: 2, 20) circle (\widthIC);
        \filldraw[draw=black,fill=white] 
          (axis cs: 2, 0) circle (\widthIC);
        \filldraw[\colorIC] 
          (axis cs: 0,10) circle (\widthIC);
        \filldraw[black] 
          (axis cs: 6, -20) circle (\widthIC);
        \draw[dotted]
          (axis cs: 6,-20) -- (axis cs: 6,20);
        \addplot [\colorTraj,\widthTraj,domain=6:8] 
          { 80-10*x };
        \filldraw[draw=black,fill=white] 
          (axis cs: 6, 0) circle (\widthIC);
        \filldraw[draw=black,fill=white] 
          (axis cs: 6, 20) circle (\widthIC);
  \end{axis}
\end{tikzpicture}
\caption{$b = 0$}

\end{subfigure}
\hfill
\begin{subfigure}[b]{\widthSubFigBB}
\centering

\begin{tikzpicture}
  [ trim axis left 
  ]
  \begin{axis} 
      [ width = \widthSubFigBBAxis
      , height = 3.5cm
      , xlabel = $t$
      , ylabel = $h$
      , y label style={
          at={(axis description cs:0.15,1)},rotate=-90,anchor=north
        }
      , axis x line = middle
      , axis y line = left
      , xmax = 4.2
      , ymin = 0.0
      , ymax = \ymaxBBh
      , xtick = {0,2,4}
      , ytick = {0}
      , yticklabels={0}
      , enlarge x limits=true
      , enlarge y limits=true
      ]
        \addplot [\colorTraj,\widthTraj,domain=0:2] 
          { -5*x*(x-2) };
        \addplot [\colorZeno,\widthTraj,domain=2:4] 
          { -5*(x-2)*(x-6) };
        \filldraw[\colorZeno] 
          (axis cs: 2,0) circle (\widthIC);
        \filldraw[\colorIC] 
          (axis cs: 0,0) circle (\widthIC);
  \end{axis}
  \begin{axis} 
      [ width = \widthSubFigBBAxis
      , height = 3.5cm
      , xlabel = $t$
      , ylabel = $v$
      , y label style={
          at={(axis description cs:0.15,1)},rotate=-90,anchor=north
        }
      , yshift = -2.5cm
      , axis x line = middle
      , axis y line = left
      , xmax = 4.2
      , ymax = \ymaxBBv
      , ymin = \yminBBv
      , xtick = {0,2,4}
      , xticklabels={}
      , ytick = {10,20}
      , yticklabels={$v_0$,$V$}
      , enlarge x limits=true
      , enlarge y limits=true
      ]
        \addplot [\colorTraj,\widthTraj,domain=0:2] 
          { 10-10*x };
        \addplot [\colorZeno,\widthTraj,domain=2:4] 
          { 40-10*x };
        \draw[dotted] 
          (axis cs: 2,-10) -- (axis cs: 2,20);
        \foreach \exponent in 
            { 1,2,3,4,5,6,7,8,9
            , 10,11,12,13,14,15,16,17,18,19
            , 20,21,22,23,24,25,26,27,28,29
            , 30,31,32,33,34,35,36,37,38,39
            , 40,41,42,43,44,45,46,47,48,49
            } {
          \edef\reset{\noexpand
            \filldraw[draw=black,fill=white] 
              (axis cs: 2,-10*0.6561^\exponent) circle (\widthIC);
          }
          \reset
        }
        \filldraw[\colorZeno] 
          (axis cs: 2, 0) circle (\widthIC);
        \filldraw[draw=black,fill=white] 
          (axis cs: 2, 20) circle (\widthIC);
        \filldraw[black] 
          (axis cs: 2, -10) circle (\widthIC);
        \filldraw[\colorIC] 
          (axis cs: 0,10) circle (\widthIC);
        \filldraw[black] 
          (axis cs: 6, -20) circle (\widthIC);
        \draw[dotted]
          (axis cs: 6,-20) -- (axis cs: 6,20);
        \filldraw[draw=black,fill=white] 
          (axis cs: 6, 0) circle (\widthIC);
        \filldraw[draw=black,fill=white] 
          (axis cs: 6, 20) circle (\widthIC);
  \end{axis}
\end{tikzpicture}
\caption{$0 < b < 1$}

\end{subfigure}

\caption{Trajectories of $\hs_B$ (Bouncing ball).
All trajectories start from $(h=0,v=v_0)$.}
\label{fig:bb-trajectories}
\end{figure}
\begin{figure}[t]
\centering
\begin{subfigure}[b]{0.3\textwidth}
\begin{tikzpicture}

\begin{axis} 
    [ width = 5cm
    , height = 3.5cm
    , xlabel = $v$
    , ylabel =  {\hspace{-1.8em}$h$}
    , axis x line = middle
    , axis y line = middle
    , xmax = 8.5
    , ymax = 390
    , ymin = 0.0
    , xtick = {-2,2}
    , xticklabels = {$-v_0$,$v_0$}
    , ytick = {0}
    , enlarge x limits=true
    , enlarge y limits=true
    ]
    \foreach \zl/\zr in 
        { 0.0/2.0
        , 2.0/6.0
        } {
      \edef\traj{\noexpand
        \addplot [\colorTraj,\widthTraj,domain=(\zr-\zl):-(\zr-\zl)] 
          ({-x},{ -5*(x-(\zr-\zl))*(x+(\zr-\zl)) });
      }
      \traj
    }
    \addplot 
      [ \colorTraj
      , \widthTraj
      , domain=-8:8]
      ({-x},{ -5*(x-8)*(x+8) });
    \draw[dashed] 
      (axis cs: 4,0) -- (axis cs: 4,330);
    \draw[dashed] 
      (axis cs: -4,0) -- (axis cs: -4,330);
    \draw[draw=none,fill=none] 
      (axis cs: -4.5,365) node 
      {\scriptsize $bV$};
    \draw[draw=none,fill=none] 
      (axis cs: 4.5,365) node 
      {\scriptsize $-bV$};
    \filldraw[draw=black,fill=white] 
      (axis cs: 4, 0) circle (\widthIC);
    \filldraw[draw=black,fill=white] 
      (axis cs: 8, 0) circle (\widthIC);
    \filldraw[black] 
      (axis cs: -2, 0) circle (\widthIC);
    \filldraw[black] 
      (axis cs: -4, 0) circle (\widthIC);
    \filldraw[\colorIC] 
      (axis cs: 2,0) circle (\widthIC);
\end{axis}

\end{tikzpicture}

\caption{$b < -1$}
\end{subfigure}
\hfill
\begin{subfigure}[b]{0.3\textwidth}

\begin{tikzpicture}

\begin{axis} 
    [ width = 5cm
    , height = 3.5cm
    , xlabel = $v$
    , ylabel =  {\hspace{-1.8em}$h$}
    , axis x line = middle
    , axis y line = middle
    , ymax = 4.2
    , ymin = 0.0
    , xtick = {-0.531441,-0.81,0.531441,0.81}
    , xticklabels = {$-v_0$,$-V$,$v_0$,$V$}
    , ytick = {0}
    , enlarge x limits=true
    , enlarge y limits=true
    ]
    \foreach \zl/\zr in 
        { 3.6855900000000004/4.217031
        , 6.712320754503902/6.941088679053512
        } {
      \edef\traj{\noexpand
        \addplot [\colorTraj,\widthTraj,domain=(\zr-\zl):-(\zr-\zl)] 
          ({-x},{ -5*(x-(\zr-\zl))*(x+(\zr-\zl)) });
      }
      \traj;
    }
    \foreach \zl/\zr in 
        { 0.9/1.71
        , 5.12579511/5.5132155990000005
        } {
      \edef\traj{\noexpand
        \addplot [\colorZeno,\widthTraj,domain=(\zr-\zl):-(\zr-\zl)] 
          ({-x},{ -5*(x-(\zr-\zl))*(x+(\zr-\zl)) });
      }
      \traj
    }
    \draw[draw=none,fill=none] 
      (axis cs: 0.9^14,4.2) node 
      {\scriptsize $-bv_0$};
    \draw[dashed] 
      (axis cs: 0.9^14,0) -- (axis cs: 0.9^14,3.8);
    \draw[draw=none,fill=none] 
      (axis cs: 0.9^9+0.05,3.4) node 
      {\scriptsize $-bV$};
    \draw[dashed] 
      (axis cs: 0.9^9,0) -- (axis cs: 0.9^9,3.15);
    \filldraw[\colorZeno] 
      (axis cs: 0,0) circle (\widthIC);
    \filldraw[draw=black,fill=white] 
      (axis cs: 0.9^14, 0) circle (\widthIC);
    \filldraw[draw=black,fill=white] 
      (axis cs: 0.9^9, 0) circle (\widthIC);
    \filldraw[draw=black,fill=white] 
      (axis cs: 0.9^2, 0) circle (\widthIC);
    \filldraw[black] 
      (axis cs: -0.9^2, 0) circle (\widthIC);
    \filldraw[black] 
      (axis cs: -0.9^6, 0) circle (\widthIC);
    \filldraw[\colorIC] 
      (axis cs: 0.9^6,0) circle (\widthIC);
\end{axis}

\end{tikzpicture}

\caption{$-1 < b < 0$}
\end{subfigure}
\hfill
\begin{subfigure}[b]{0.3\textwidth}

\begin{tikzpicture}

\begin{axis} 
    [ width = 5cm
    , height = 3.5cm
    , xlabel = $v$
    , ylabel =  {\hspace{-1.8em}$h$}
    , axis x line = middle
    , axis y line = middle
    , ymax = 4
    , ymin = 0.0
    , xtick = {-0.531441,-0.81,0.531441,0.81}
    , xticklabels = {$-v_0$,$-V$,$v_0$,$V$}
    , ytick = {0}
    , enlarge x limits=true
    , enlarge y limits=true
    ]
    \foreach \zl/\zr in 
        { 3.6855900000000004/4.217031
        } {
      \edef\traj{\noexpand
        \addplot [\colorTraj,\widthTraj,domain=(\zr-\zl):-(\zr-\zl)] 
          ({-x},{ -5*(x-(\zr-\zl))*(x+(\zr-\zl)) });
      }
      \traj
    }
    \foreach \zl/\zr in 
        { 0.9/1.71
        } {
      \edef\traj{\noexpand
        \addplot [\colorZeno,\widthTraj,domain=(\zr-\zl):-(\zr-\zl)] 
          ({-x},{ -5*(x-(\zr-\zl))*(x+(\zr-\zl)) });
      }
      \traj
    }
    \foreach \exponent in 
        { 7,8,9
        , 10,11,12,13,14,15,16,17,18,19
        , 20,21,22,23,24,25,26,27,28,29
        , 30,31,32,33,34,35,36,37,38,39
        , 40,41,42,43,44,45,46,47,48,49
        } {
      \edef\reset{\noexpand
        \filldraw[draw=black,fill=white] 
          (axis cs: -0.9^\exponent, 0) circle (\widthIC);
      }
      \reset
    }
    \foreach \exponent in 
        { 3,4,5,6,7,8,9
        , 10,11,12,13,14,15,16,17,18,19
        , 20,21,22,23,24,25,26,27,28,29
        , 30,31,32,33,34,35,36,37,38,39
        , 40,41,42,43,44,45,46,47,48,49
        } {
      \edef\reset{\noexpand
        \filldraw[draw=black,fill=white] 
          (axis cs: -0.9^\exponent, 0) circle (\widthIC);
      }
      \reset
    }
    \draw[dashed] 
      (axis cs: -0.9^7,0) -- (axis cs: -0.9^7,3.0);
    \draw[draw=none,fill=none] 
      (axis cs: -0.9^7-0.1,3.3) node 
      {\scriptsize $-bv_0$};
    \draw[dashed] 
      (axis cs: -0.9^3,0) -- (axis cs: -0.9^3,2.0);
    \draw[draw=none,fill=none] 
      (axis cs: -0.9^3-0.1,2.2) node 
      {\scriptsize $-bV$};
    \filldraw[\colorIC] 
      (axis cs: 0.9^6,0) circle (\widthIC);
    \filldraw[black] 
      (axis cs: -0.9^6, 0) circle (\widthIC);
    \filldraw[black] 
      (axis cs: -0.9^2, 0) circle (\widthIC);
    \filldraw[draw=black,fill=white] 
      (axis cs: 0.9^2, 0) circle (\widthIC);
    \filldraw[\colorZeno] 
      (axis cs: 0,0) circle (\widthIC);
\end{axis}

\end{tikzpicture}

\caption{$0 < b < 1$}
\end{subfigure}

\caption{Set of reachable states of $\hs_B$ (Bouncing ball).
The set $I$ is always $\{(h=0,v=v_0)\}$.
In case (a) there is an expanding sequence of parabolas \textcolor{\colorTraj}{$(-b)^n v_0\to+\infty$}.
In case (b) there are two shrinking sequences of parabolas \textcolor{\colorTraj}{$(-b)^n v_0\to 0$} and \textcolor{\colorZeno}{$(-b)^n V\to 0$}.}
\label{fig:bb-states}
\end{figure}
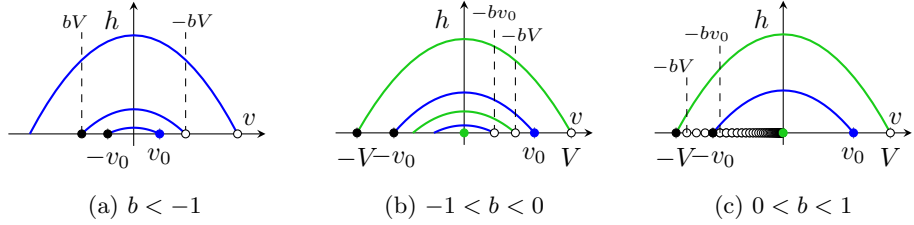
\begin{enumerate}
\item $S_f=S_s=\bigcup_n S(b^nv_0)$ when $b<-1$
\item $S_f=S_s=S(v_0)$ when $b=-1$ (elastic bounce)
\item $S_f=\bigcup_n S(b^nv_0)\subset S_f\cup S(0)\cup(\bigcup_n S(b^nV))=S_s$ when $-1<b<0$
\item $S_f=S_s=S(v_0)\cup S(0)\cup S(V)$ when $b=0$ (inelastic bounce)
\item $S_f=S(v_0)\cup S'(b,v_0)\subset S_f\cup S(0)\cup S(V)\cup S'(b,V)=S_s$ when $0<b<1$, with $S'(b,u)\defeq\{(0,-b^nu)|0\leq n\}$ sequence of instantaneous slowdowns
\item $S_f=S_s=S(v_0)$ when $b=1$
\item $S_f=S_s=S(v_0)\cup S'(b,v_0)$ when $1<b$, now $S'(b,u)$ is a sequence of instantaneous accelerations.
\end{enumerate}
To make $\hs_B$ compact, the simplest is to put an upper bound to the energy of the system, say $E_0=E(0,V_0)$ with $V_0>V$, and allow only $b$ st $|b|\leq 1$, so that the energy cannot increase when the ball bounces.  So we replace $\hs_B$ with the following compact HS $\hs$, whose support is $S_0=\{(h,v)|0\leq h\land E(h,v)\leq E_0\}$
\begin{itemize}
    \item $F=\{((h,v),(\Der{h},\Der{v}))|(h,v):S_0\land\Der{h}=v\land\Der{v}=-1\}$
    \item $G=\{((0,v),(0,v^+))|0\leq -v\leq V_0\land v^+=b*v\}\uplus\{((0,0),(0,V))\}$.
\end{itemize}
To define $S_R$ we fix a compact HS $\hs_0$ with support $S_0$.  The simplest choice is to
take $F_0=F$ and replace $G$ with a $G_0$ independent from $b$, namely
\begin{itemize}
    \item $G_0=\{((0,v),(0,v^+))|0\leq -v\leq V_0\land |v^+|\leq-v\}\uplus\{((0,0),(0,V))\}$.
\end{itemize}
The relations among $S_s$, $S_r$ and $S_R$, when $|b|\leq 1$ and $0<v_0<V<V_0$, are
\begin{enumerate}
\item $S_s=S_r=S(v_0)\subset \bigcup\{S(v)|v:[0,V]\}=S_R$ when $b=-1$ (elastic bounce)
\item $S_f=S_r=S_R=\bigcup_n S(b^nv_0)\cup S(0)\cup(\bigcup_n S(b^nV))$ when $-1<b<0$
\item $S_s=S_r=S_R=S(v_0)\cup S(0)\cup S(V)$ when $b=0$ (inelastic bounce)
\item $S_s=S_r=S_R=S(v_0)\cup S'(b,v_0)\cup S(0)\cup S(V)\cup S'(b,V)$ when $0<b<1$
\item $S_s=S_r=S(v_0)\subset S(v_0)\cup S(V)\cup\{(0,v)|-v:[0,V]\}=S_R$ when $b=1$.
\end{enumerate}
There is an informal explanation for $S_R$ in the case $b=-1$ (elastic bounce).
After each bounce the ball may lose a bit of energy, thus after sufficiently many bounces it may stop (minimum energy).
After a kick the energy will reach the maximum value allowed, and then it may decrease again after each bounce.
Thus any level of energy in $[0,E(0,V)]$ is reachable, assuming $0<v_0<V$.
More formally we define $\hs_\delta\ll\hs$ in $\HS_c(\hs_0)$ st $\hs_\delta\to\hs$ when $\delta\to 0$.
Since $\hs_0$ allows only perturbations in $G$, we define $G_\delta$ (for $\delta>0$) as
\[\{((0,v),(0,v^+))|0\leq -v\leq V_0\land|v^+|\leq-v\land -v-\delta\leq v^+\}\uplus\{((0,0),(0,V))\}\]
When $-v$ is small, ie $0\leq-v\leq\delta\leq V$, $|v^+|\leq-v$ and the energy is $\leq\frac{\delta^2}{2}$.
Otherwise, $0<\delta<-v$, $v^+:[-v-\delta,-v]$ and the energy loss is $\leq V\delta-\frac{\delta^2}{2}$.

\begin{remark}\label{rmk-constrain}
It is important what $\hs_0$ is chosen to capture the hard constrains on the HS $\hs$ of interest,
because it can affect how $S_R$ is computed.
For instance, for the bouncing ball one may replace $\hs_0$ with a \emph{more relaxed} $\hs'_0$
\begin{itemize}
    \item $G'_0=\{((0,v),(0,v^+))|0\leq -v\leq V_0\land |v^+|\leq V_0\}\uplus\{((0,0),(0,V))\}$
\end{itemize}
$\hs'_0$ has the same support of $\hs_0$, but $\hs'_0<\hs_0$, because after a bounce the ball can increase its energy as far as it stays below the upper bound $E(0,V_0)$.  This change results in a bigger subset $S_R$ when $|b|=1$, namely
\begin{itemize}
\item $S_0=S_R$ when $b=-1$, ie any state in the support of $\hs$ is reachable because of the more permissive perturbations
\item $S(v_0)\cup S(V)\cup\{(0,v)|-v:[0,V_0]\} =S_R$ when $b=1$.
\end{itemize}
\end{remark}

\section*{Conclusions and Future Work}

The main contributions of this paper concern reachability analysis in the context of hybrid (and continuous) systems.

Firstly, we have proposed safe reachability $\Rs{\hs}(I)$, which computes an over-approximation of the set of states reachable in finite time from the set $I$ of initial states by the hybrid system $\hs$, and compared it with the more naive reachability $\Rf{\hs}(I)$, which computes only an under-approximation.

Secondly, and more importantly, we have addressed the issue of \emph{robustness} of an analysis $A$ cast as a monotonic map $A:\Po(X,Y)$ between complete lattices $X$ and $Y$ of a particular form (ie, hyperspaces of metric spaces).  Robustness of $A$ means that $A(x_\delta)\to A(x)$ as $\delta\to 0$, where $x_\delta$ is a \emph{small perturbation} of $x$ depending on a $\delta>0$, which measures the level of inaccuracy.  In some cases (ie, when the metric spaces are compact) robustness amounts to Scott continuity, and one can exploit the following facts:
\begin{itemize}
\item Every monotonic map $A:\Po(X,Y)$ between complete lattices has a best Scott continuous approximation $\bca{A}\leq A:\Po(X,Y)$.
\item When $X$ is a continuous lattice, $\bca{A}(x)$ is the sup of $\{A(b)|b\ll_X x\}$, ie, it is computed by applying $A$ to \emph{way-below} approximations of $x$. 
\end{itemize} 
While the importance of safe/sound analyses is widely recognized, the issue of robustness is mostly overlooked (one reason being that for discrete systems it is not an issue).
In our view, robustness has at least two immediate implications:
\begin{description}
\item[Modeling languages.] There should be syntactic support to distinguish between hard and soft constraints on a hybrid system $\hs$.  Hard constraints must be satisfied also by small perturbations $\hs_\delta$. Thus, they identify the complete lattice (hyperspace) $X$ where $\hs$ is placed, while soft constraints provide the additional information to identify $\hs$ uniquely within $X$.  The distinction would be needed by tools that implement a robust analysis and can be ignored by other tools.
In \cite{gao2013delta} there is no explicit annotation for soft constraints, instead there is a re-interpretation of logical formula, which injects $\delta$-noise in specific sub-formulas.

\item[Finite model checking.]
Counterexample-guided Abstraction \& Refinement (CEGAR) is a general way of analyzing a system $\hs$ with an infinite state space by
leveraging finite model checking tools (see \cite{clarke2000counterexample,clarke2003verification}).
In the setting of abstract interpretation, CEGAR amounts to approximating an analysis $A:\Po(X,X)$ with \emph{finite} analyses $A':\Po(X_f,X_f)$, ie:
\[\commdiag{
X&&\rTo~{A}&&X\\
\uIntoB~\gamma&&\leq&&\uIntoB~\gamma&&\mbox{$\gamma$ sub-lattice inclusion}\\
X_f&&\rTo~{A'}&&X_f&&\mbox{$X_f$ finite (complete) lattice}\\
}\]
Whenever $X_f$ is fixed, there is a best approximation $A_f:\Po(X_f,X_f)$ of $A$ given by $A_f\defeq \gamma^R\circ A\circ\gamma$.

If $\hs$ is a HS on $\State=\Real^n$ and its support is included in a compact subset $K$, then there are three reachability analyses $\bca{\Rs{\hs}}\leq\Rs{\hs}\leq\Rf{\hs}:\Po(X,X)$, where $X=\cL(K)$ is a continuous lattice with a countable base, but uncountably many elements (unless $K$ is finite).  One may wonder whether replacing $X$ with a finite sub-lattice $X_f$ would make the three analyses indistinguishable: the answer is no (counterexamples can be given using $\hs_E$ and $\hs_D$ in Examples~\ref{ex-expand}~and~\ref{ex-decay}).

However, there is a way to turn a finite approximation $A' : \Po(X_f,X_f)$ of $A:\Po(X,X)$ into a finite approximation $A''$ of $\bca{A}$, provided that $X$ is a continuous lattice, namely, $A''(x)\defeq A'(x_b)$ where $x_b$ is the biggest element in $X_f$ such that $\gamma(b)\ll_X\gamma(x)$.
\end{description}
As future work we plan to address \textbf{computability} issues.  More specifically, given a compact HS $\hs_0$ on $\State$ with support $S_0$, is $\bca{\Rs{}}:\Po(\HS_c(\hs_0)\X\cL(S_0),\cL(S_0))$ computable?
When $\State$ has a countable dense subset, all  continuous lattices involved have a countable base, and the question can be formulated in a well-established setting.
In \cite{EP06} the authors study computability of the evolution map for a \emph{compact and continuous flow automaton} (ccFA for short).  In our setting, a ccFA is a tuple $(C,f,G,I)$, with $C$ a compact subset of $\Real^n$ called invariant, $f:\Top(C,\Real^n)$ the flow function, $G:\CS(C\times C)$ the jump relation, and $I:\CS(C)$ the set of initial states.
The graph of $f$ is flow relation $f:\CS(C\times V)$, with $V=f(C)$. Therefore, $(f,G)$ is a compact HS on $\Real^n$ with support included in the compact subset $K=C\cup V$.

The main contributions in \cite{EP06} are: the definition of the denotational semantics $\Den{(C,f,G,I)}:\Top(\Time,\cL(C))\cong\CS(\Time\X C)$ of a ccFA $(C,f,G,I)$, which is \emph{computable} for \emph{effectively given} ccFA (cf. \cite[Thm~38]{EP06}); and computational adequacy, ie $\Den{C,f,G,I}=\Ef{(f,G)}(I)$, for \emph{separated} ccFA (cf. \cite[Thm~25]{EP06}).

Most of the steps in defining the denotational semantics of a ccFA use Scott continuous maps, but for the lack of a continuous lattice/domain of ccFA.  It would be interesting to see if their denotational semantics extends to compact HS, giving
a Scott continuous \emph{computable} map $\Den{-}:\Po(\HS_c(\hs_0)\X\cL(S_0),\cL(S_0))$, and then compare it with the map $\bca{\Rs{}}$ between the same lattices.

\bibliographystyle{alpha}
\bibliography{local}

\newcommand{\etalchar}[1]{$^{#1}$}
\begin{thebibliography}{GHK{\etalchar{+}}03}

\bibitem[ACH{\etalchar{+}}95]{ACHH95}
Rajeev Alur, Costas Courcoubetis, Nicolas Halbwachs, Thomas~A Henzinger, P-H
  Ho, Xavier Nicollin, Alfredo Olivero, Joseph Sifakis, and Sergio Yovine.
\newblock The algorithmic analysis of hybrid systems.
\newblock {\em Theoretical computer science}, 138(1):3--34, 1995.

\bibitem[AJ94]{AbramskyJung94-DT}
S.~Abramsky and A.~Jung.
\newblock Domain theory.
\newblock In S.~Abramsky, D.~M. Gabbay, and T.~S.~E. Maibaum, editors, {\em
  Handbook of Logic in Computer Science}, volume~3, pages 1--168. Clarendon
  Press, Oxford, 1994.

\bibitem[CC92]{cousot1992abstract}
Patrick Cousot and Radhia Cousot.
\newblock Abstract interpretation frameworks.
\newblock {\em Journal of logic and computation}, 2(4):511--547, 1992.

\bibitem[CFH{\etalchar{+}}03]{clarke2003verification}
Edmund Clarke, Ansgar Fehnker, Zhi Han, Bruce Krogh, Olaf Stursberg, and
  Michael Theobald.
\newblock Verification of hybrid systems based on counterexample-guided
  abstraction refinement.
\newblock In {\em TACAS}, volume~3, pages 192--207. Springer, 2003.

\bibitem[CGJ{\etalchar{+}}00]{clarke2000counterexample}
Edmund Clarke, Orna Grumberg, Somesh Jha, Yuan Lu, and Helmut Veith.
\newblock Counterexample-guided abstraction refinement.
\newblock In {\em Computer aided verification}, pages 154--169. Springer, 2000.

\bibitem[Con90]{Conway:Fun_Analysis:2ed:1990}
John~B. Conway.
\newblock {\em A Course in Functional Analysis}.
\newblock Springer, 2nd edition edition, 1990.

\bibitem[CR04]{cuijpers2004}
Pieter Jan~Laurens Cuijpers and Michel~Adriaan Reniers.
\newblock Topological (bi-) simulation.
\newblock {\em Electronic Notes in Theoretical Computer Science}, 100:49--64,
  2004.

\bibitem[Cui07]{cuijpers2007}
Pieter J.~L. Cuijpers.
\newblock On bicontinuous bisimulation and the preservation of stability.
\newblock In Alberto Bemporad, Antonio Bicchi, and Giorgio Buttazzo, editors,
  {\em Hybrid Systems: Computation and Control}, volume 4416 of {\em Lecture
  Notes in Computer Science}, pages 676--679. Springer, 2007.

\bibitem[Dur16]{AdamPhD}
Adam Duracz.
\newblock {\em Rigorous Simulation: Its Theory and Applications}.
\newblock PhD thesis, Halmstad University Press, 2016.

\bibitem[Eda95]{Edalat95}
Abbas Edalat.
\newblock Dynamical systems, measures and fractals via domain theory.
\newblock {\em Information and Computation}, 120(1):32--48, July 1995.

\bibitem[EP07]{EP06}
Abbas Edalat and Dirk Pattinson.
\newblock Denotational semantics of hybrid automata.
\newblock {\em The Journal of Logic and Algebraic Programming}, 73(1):3--21,
  2007.

\bibitem[Fr{\"a}99]{franzle1999analysis}
Martin Fr{\"a}nzle.
\newblock Analysis of hybrid systems: An ounce of realism can save an infinity
  of states.
\newblock In {\em Computer Science Logic}, pages 126--139. Springer, 1999.

\bibitem[GHK{\etalchar{+}}03]{Gierz-ContinuousLattices-2003}
G.~Gierz, K.~H. Hofmann, K.~Keimel, J.~D. Lawson, M.~W. Mislove, and D.~S.
  Scott.
\newblock {\em Continuous Lattices and Domains}, volume~93 of {\em Encycloedia
  of Mathematics and its Applications}.
\newblock Cambridge University Press, 2003.

\bibitem[GKC13]{gao2013delta}
Sicun Gao, Soonho Kong, and Edmund~M Clarke.
\newblock Delta-complete reachability analysis (part 1).
\newblock Technical report, Carnegie-Mellon Univ. School of Computer Science,
  2013.

\bibitem[GST09]{GST2009}
Rafal Goebel, Ricardo~G Sanfelice, and A~Teel.
\newblock Hybrid dynamical systems.
\newblock {\em Control Systems, IEEE}, 29(2):28--93, 2009.

\bibitem[Kel75]{Kel1975}
John~L Kelley.
\newblock {\em General topology}.
\newblock Springer, 1975.

\bibitem[Kel82]{kelly1982basic}
Gregory~Maxwell Kelly.
\newblock {\em Basic concepts of enriched category theory}.
\newblock Number~64 in Lecture Notes in Mathematics. CUP Archive, 1982.

\bibitem[Pla08]{P08}
Andr{\'e} Platzer.
\newblock Differential dynamic logic for hybrid systems.
\newblock {\em Journal of Automated Reasoning}, 41(2):143--189, 2008.

\end{thebibliography}

\appendix
\section{Proofs}\label{sec-proofs}

\begin{proof}[\textbf{of Thm~\ref{thm-fRE}}]
$\Ef{}(I)$ and $\Rf{}(I)$ are the least prefix-points of some monotonic maps on complete lattices of the form $(\PS(\State),\subseteq)$. Thus, we can exploit the universal property of the least prefix-point $X$ for a monotonic map $F$, ie, $F(Y)\leq Y\implies X\leq Y$.
\begin{enumerate}
\item Consider the monotonic maps $F$ and $F_I$ on $\PS(\Time\X\State)$:
\[F(S) \defeq \{(t+d,s')|\exists s.(t,s):S\land s\rTo^ds'\},\quad F_I(S) \defeq (\{0\}\X I)\cup F(S).\]
$\Ef{}(I)$ is the least prefix-point of $F_I$. Since $F_{I_0}(S)\subseteq F_{I_1}(S)$ when $I_0\subseteq I_1$, a prefix-point for $F_{I_1}$ is also a prefix-point for $F_{I_0}$. Hence, we conclude that $\Ef{}(I_0)\subseteq\Ef{}(I_1)$.

Since $I\subseteq\cup K$ when $I:K$, then $\Ef{}(I)\subseteq\Ef{}(\cup K)$. Now, let $U=\{\Ef{}(I)|I:K\}$. To prove $\Ef{}(\cup K)\subseteq\cup U$, observe that:
\begin{itemize}
\item $F$ preserves unions, thus $F(\cup U)=\cup\{F(S)|S:U\}$;
\item $\forall S:U.F(S)\subseteq S$, thus $F(\cup U)=\cup U$;
\item $\forall I:K.\exists S:U.\{0\}\X I\subseteq S$, thus $\{0\}\X(\cup K)\subseteq\cup U$.
\end{itemize}
Therefore, $\cup U$ is a prefix-point for $F_{\cup K}$.

\item Consider the monotonic maps $G$ and $G_I$ on $\PS(\State)$:
\[G(S) \defeq \{s'|\exists s:S.s\rTo s'\}, \quad G_I(S) \defeq I\cup G(S).\]
$\Rf{}(I)$ is the least prefix-point of $G_I$.
By analogy with the previous point, one can prove that $\Rf{}$ is monotonic and preserves unions.

Let $S$ be a prefix-point of $G_I$, ie, $G_I(S)\subseteq S$. Then:
\begin{itemize}
\item $I\subseteq S$, because $I\subseteq G_I(S)$.
\item $G_S(S)\subseteq S$, because $G(S)\subseteq G_I(S)$ and $G_S(S)=S\cup G(S)$.
\end{itemize}
By taking $S=\Rf{}(I)$, we conclude $I\subseteq\Rf{}(I)\subseteq\Rf{}(\Rf{}(I))\subseteq\Rf{}(I)$.

To prove $\pi(\Ef{}(I))=\Rf{}(I)$, observe that:
\begin{itemize}
\item $\forall E:\PS(\Time\X\State).\pi(F_I(E))=G_I(\pi(E))$. Hence, $\pi(\Ef{}(I))\supseteq\Rf{}(I)$.
\item $\forall S:\PS(\State).G_I(S)\subseteq S\implies F_I(\Time\X S)\subseteq\Time\X S$.
\end{itemize}
Therefore, $\Ef{}(I)\subseteq\Time\X\Rf{}(I)$, and consequently, $\pi(\Ef{}(I))\subseteq\Rf{}(I)$.

\item Consider the monotonic maps $F_I$ on $\PS(\Time\X\State)$, $G_I$ on $\PS(\State)$, $G^t_J$ on $\PS(\Real\X\State)$, whose least prefix-points are
$\Ef{\hs}(I)$, $\Rf{\hs}(I)$, and $\Rf{t(\hs)}(J)$, respectively.
Since $(t,s)\rTo_{t(\hs)}(t+d,s')\iff 0\leq d\land s\rTo_\hs^d s'$, by Prop~\ref{thm-HT}, these maps
are related as follows:
\begin{itemize}
\item $\forall E:\PS(\Time\X\State).F_I(E)=G^t_{\{0\}\X I}(E)$, so $\Ef{\hs}(I)=(\Rf{t(\hs)}(\{0\}\X I)$.
\item $\forall S:\PS(\Real\X\State).\pi(G^t_J(S))=G_{\pi(J)}(\pi(S))$, so $\pi(\Rf{t(\hs)}(J))\supseteq\Rf{\hs}(\pi(J))$.
\item $\forall S:\PS(\State).G_{\pi(J)}(S)\subseteq S\implies G^t_J(\Real\X S)\subseteq\Real\X S$.
\end{itemize}
As a result, $\Rf{t(\hs)}(J)\subseteq\Real\X\Rf{\hs}(\pi(J))$, and
$\pi(\Rf{t(\hs)}(J))\subseteq\Rf{\hs}(\pi(J))$.
\qed
\end{enumerate}
\end{proof}

\begin{proof}[\textbf{of Thm~\ref{thm-sRE}}]
$\Es{}(I)$ and $\Rs{}(I)$ are defined as least prefix-points of monotonic maps on complete lattices of the form $(\CS(\State),\subseteq)$,
$\CS(\State)$ is closed wrt arbitrary intersections and finite unions computed in $\PS(\State)$, and
the monotonic map $S\mapsto\cl{S}$ from $\PS(\State)$ to $\CS(\State)$ preserves finite unions.
\begin{enumerate}
\item $\Es{}(I)$ is the least prefix-point of a monotonic map $F'_I$ on $\CS(\Time\X\State)$ given by:
\[F'_I(S) \defeq (\{0\}\X I)\cup\cl{\{(t+d,s')|\exists s.(t,s):S\land s\rTo^ds'\}},\]
and the properties of $\Es{}$ are proved similar to those of $\Ef{}$, except
for the need to restrict to finite unions.

\item $\Rs{}(I)$ is the least prefix-point of a monotonic map $G'_I$ on $\CS(\State)$ given by:
\[G'_I(S) \defeq I\cup\cl{\{s'|\exists s:S.s\rTo s'\}},\]
and the properties of $\Rs{}$ are proved by analogy with those of $\Rf{}$. In particular, $\pi(\Es{}(I))\subseteq\Rs{}(I)$ follows from:
\begin{equation*}
    \forall S:\CS(\State).G'_I(S)\subseteq S\implies F'_I(\Time\X S)\subseteq\Time\X S.
\end{equation*}

As a result, $\Es{}(I)\subseteq\Time\X\Rs{}(I)$, and consequently $\pi(\Es{}(I))\subseteq\Rs{}(I)$. 

\item If $I:\PS(\State)$ and $S:\PS(\Time\X\State)$, then $I\subseteq\cl{I}:\CS(\State)$ and
$F_I(S)\subseteq F'_{\cl{I}}(S):\CS(\Time\X\State)$.
Hence, $\Ef{}(I)\subseteq\Es{}(\cl{I}):\CS(\Time\X\State)$, and consequently
$\Ef{}(I)\subseteq\Es{}(\cl{I})$.

The inclusion $\Rf{}(I)\subseteq\Rs{}(\cl{I})$ follows from $\Ef{}(I)\subseteq\Es{}(\cl{I})$,
since:
\begin{itemize}
\item $\Rf{}(I)=\pi(\Ef{}(I))$, by Thm~\ref{thm-fRE}.
\item $\pi(\Es{}(\cl{I}))\subseteq\Rs{}(\cl{I})$, by the previous point.
\end{itemize}

\item Consider the monotonic maps $F'_I$ on $\CS(\Time\X\State)$, $G'_I$ on $\CS(\State)$, and $G^t_J$ on $\CS(\Real\X\State)$, whose least prefix-points are
$\Es{\hs}(I)$, $\Rs{\hs}(I)$ and $\Rs{t(\hs)}(J)$, respectively.
Since $(t,s)\rTo_{t(\hs)}(t+d,s')\iff 0\leq d\land s\rTo_\hs^d s'$, by Prop~\ref{thm-HT}, these maps
are related as follows:
\begin{itemize}
\item $\forall E:\CS(\Time\X\State).F'_I(E)=G^t_{\{0\}\X I}(E)$, so $\Es{\hs}(I)=(\Rs{t(\hs)}(\{0\}\X I)$.
\item $\forall S:\CS(\State).G'_{\cl{\pi(J)}}(S)\subseteq S\implies G^t_J(\Real\X S)\subseteq\Real\X S$.
\end{itemize}

Hence, $\Rs{t(\hs)}(J)\subseteq\Real\X\Rs{\hs}(\cl{\pi(J)})$, and
$\pi(\Rs{t(\hs)}(J))\subseteq\Rs{\hs}(\cl{\pi(J)})$.
\qed
\end{enumerate}
\end{proof}

\begin{proof}[\textbf{of Thm~\ref{thm-mon-maps}}]
A Banach space is a metric space, and one can rely on characterizations of topological notions or exploit properties that are specific to metric spaces, eg, $x:\cl{U}\iff$ $x$ is the limit of a sequence $(x_n|n:\omega)$ in $U$.
\begin{enumerate}
    \item $t(F,G)=(F',G')$, where $F'=C_1\X F$ and $G'=C_2\X G$ with $C_1=\Real\X\{1\}$ and $C_2=\{(t,t)|t:\Real\}$, 
    modulo the isomorphism $(\Real\X\State)^2\lTo\Real^2\X\State^2$. $C_1$ and $C_2$ are closed subsets of $\Real^2$. Thus, $\cl{F'}=C_1\X\cl{F}$ and $\cl{G'}=C_2\X\cl{G}$, ie, $t(\cl{\hs})=\cl{t(\hs)}$.
    
    \item $\Ts(\hs,C)\subseteq C$ means that $C\leq\Ts(\hs,C)$. Hence, $C\leq\Ts(\hs,C)\leq\Ts(\hs',S')$ when $\hs\leq\hs'$ and $C\leq S'$, because $\Ts$ is monotonic.

    For $\Rf{}$ and $\Rs{}$ one can proceed as above, by  observing that $\Ts(\hs,C)\subseteq C$ means ``$C$ is closed wrt the transition relation $\rTo_\hs$'', and therefore $C\subseteq\Rf{}(\hs,C)\subseteq\Rs{}(\hs,C)\subseteq C$, by the assumptions on $C$ and by definition of $\Rf{}(\hs,C)$ and $\Rs{}(\hs,C)$.

    \item We prove that $\cl{\Sp(\hs)}$ is closed under the transition relation $\rTo_\hs$:
    \begin{itemize}
        \item If $s\rTo^0_\hs s'$, then $\rel{G}{s}{s'}$. Therefore, $s':\Sp(\hs)\subseteq\cl{\Sp(\hs)}$.
        \item If $s\rTo^d_\hs s'$ with $d>0$, then there exists $f:\Top([0,d],\State)$ such that $f(0,d)\subseteq \Sp(\hs)$ and $s'=f(d)$. Hence, $s':\cl{\Sp(\hs)}$, because $f$ is continuous.
    \end{itemize}
    
    \item It should be clear that $\cl{\Sp(\cl{\hs})}\leq\cl{\Sp(\hs)},\Sp(\cl{\hs})\leq \Sp(\hs):\pL(\State)$, as the maps $\Sp$ and $f^R(U)=\cl{U}$ are monotonic.
    To prove $\cl{\Sp(\cl{\hs})}=\cl{\Sp(\hs)}$, it suffices to show that $\cl{\Sp(\hs)}\leq \Sp(\cl{\hs})$, since $\cl{\cl{U}}=\cl{U}$.
    \begin{itemize}
        \item $\cl{\hs}=\cl{F+G}=\cl{F}+\cl{G}$. Therefore,
        $s:\Sp(\cl{\hs})\iff$\\$\exists x.(x:\cl{F}\land s=\pi_1(x))\lor(x:\cl{G}\land s=\pi_1(x))\lor(x:\cl{G}\land s=\pi_2(x))$.
        \item Assume wlog that $x:\cl{F}\land s=\pi_1(x)$, we have to show that $s:\cl{\Sp(\hs)}$.
        \item $x:\cl{F}$ implies that there exists a sequence $(x_n|n:\omega)$ in $F$ with limit $x$.
        \item Let $s_n=\pi_1(x_n)$. Then $(s_n|n:\omega)$ is a sequence in $\Sp(\hs)$ with limit $s$, because $\pi_1$ is continuous. Thus, $s:\cl{\Sp(\hs})$.
    \end{itemize}
    
    \item 
    $\Sp(\hs)$ is the union $\pi_1(F)\cup\pi_1(G)\cup\pi_2(G)$, and projections $\pi_i:\State\X\State\rTo\State$ are continuous maps.
    The image of a compact set along a continuous map is always compact.
    Therefore, if $\hs$ is compact---ie, $F$ and $G$ are compact---then $\Sp(\hs)$ is compact, because it is the union of three compact subsets. Moreover, in metric spaces compact subsets are always closed.\qed
\end{enumerate}
\end{proof}

\begin{proof}[\textbf{of Thm~\ref{thm-robust-equal}}]
We exploit the facts (F1-F3) in the proof of Lemma~\ref{thm-top}.

(1)$\implies$(2).
Given $U_2\subseteq\CS(\State_2)$ open (for the Robust topology), we have to prove that $U_1\defeq\inv{A}(U_2)\subseteq\CS(\State_1)$ is open. $U_1$ is upward closed, because $U_2$ is upward closed and $A$ is monotonic.
If $C:U_1$, ie, $A(C):U_2$, we have to find $\delta>0$ such that$\uparrow B(C,\delta)\subseteq U_1$, ie,
$C'\subseteq B(C,\delta)\implies A(C'):U_2$:
\begin{itemize}
\item $A(C):U_2$ and $U_2$ open imply that $\uparrow B(A(C),\delta')\subseteq U_2$ for some $\delta'>0$.
\item Now, let $\epsilon$ be in $(0,\delta')$. By (F3), we get $A(C)_\epsilon\subseteq B(A(C),\delta')$.
\item By robustness of $A$, there exists $\delta>0$ such that $A(C_\delta)\subseteq A(C)_\epsilon$.
\item By monotonicity of $A$, we have $A(C')\subseteq B(A(C),\delta')$, ie, $A(C'):U_2$, when $C'\subseteq B(C,\delta)\subseteq C_\delta$.
\end{itemize}
(2)$\implies$(1). The Robust topology on $\CS(\State_i)$ is between Scott and Alexandrov topologies, which have the same specialization order, ie, reverse inclusion.  Thus, $A$ is monotonic wrt the specialization orders, because it is continuous wrt the Robust topologies.

For $C:\CS(\State_1)$ and $\epsilon>0$, let $U_2\defeq\{C':\CS(\State_2)|\exists\epsilon'>0.B(C',\epsilon')\subseteq B(A(C),\epsilon)\}$.  Clearly $A(C):U_2\subseteq\uparrow B(A(C),\epsilon)\subseteq\uparrow A(C)_\epsilon$.
We prove that $U_2$ is open for the Robust topology, ie, $\forall C':U_2.\exists\delta>0.\uparrow B(C',\delta)\subseteq U_2$:
\begin{itemize}
\item $C':U_2$ implies $B(C',\epsilon')\subseteq B(A(C),\epsilon)$ for some $\epsilon'>0$.
\item By (F1), $B(B(C',\delta),\delta)\subseteq B(C',\epsilon')$, where $\delta=\epsilon'/2$.
\item Therefore, $\uparrow B(C',\delta)\subseteq U_2$.
\end{itemize}
$U_1\defeq\inv{A}(U_2)$ is open in the Robust topology, because $U_2$ is open and $A$ is continuous.
But $C:U_1$. Therefore, there exists $\delta'>0$ such that $\uparrow B(C,\delta')\subseteq U_1$.
If $\delta$ is in $(0,\delta')$, then (F3) implies that $C_\delta\subseteq B(C,\delta')$. Hence, $C_\delta:U_1$ and $A(C_\delta):U_2\subseteq\uparrow A(C)_\epsilon$, which means that $A(C_\delta)\subseteq A(C)_\epsilon$.
\qed\end{proof}

\end{document}